\documentclass[12pt]{article}
\usepackage[utf8]{inputenc}
\usepackage{graphicx}
\usepackage{natbib}
\usepackage{color}
\usepackage{amsmath}
\usepackage{amssymb}
\usepackage{amsfonts}
\usepackage[hidelinks]{hyperref}
\usepackage{amsthm}
\newtheorem{theorem}{Theorem}
\newtheorem{corollary}{Corollary}
\newtheorem{lemma}{Lemma}
\newtheorem{proposition}{Proposition}
\newtheorem{assumption}{}

\usepackage[algo2e]{algorithm2e} 
\usepackage{algorithm}
\usepackage{verbatim}
\usepackage{caption}
\usepackage{setspace}
\usepackage{subcaption}
\usepackage{multirow}
\usepackage{lineno}
\usepackage[left=1in,right=1in,top=1in,bottom=1in]{geometry}
\usepackage{cancel}
\usepackage{dsfont}

\newcommand{\bc}{\boldsymbol c}

\newcommand{\bu}{\boldsymbol{u}}
\newcommand{\bv}{\boldsymbol{v}}

\newcommand{\btheta}{\boldsymbol{\theta}}
\newcommand{\bpsi}{\boldsymbol{\psi}}

\newcommand{\blambda}{\boldsymbol{\lambda}}
\newcommand{\bA}{{\bf A}}
\newcommand{\bB}{{\bf B}}
\newcommand{\bX}{{\bf X}}
\newcommand{\bY}{{\bf Y}}
\newcommand{\bZ}{{\bf Z}}
\newcommand{\bP}{{\bf P}}
\newcommand{\bbeta}{\boldsymbol{\beta}}

\title{Universal Inference for model selection on networks}
\author{Eric Yanchenko\footnote{Human and AI Center, Akita International University. \url{eyanchenko@aiu.ac.jp}},  Jonathan P. Williams\footnote{Department of Statistics, North Carolina State University.} and Ryan Martin$^{\dagger}$}

\begin{document}

\maketitle
\thispagestyle{empty}


\begin{abstract}
\noindent
Model selection and hypothesis testing are important tasks on networks. A key challenge lies in the inherent dependence in network data, as well as the fact that typically only a single realization is observed. As a result, many existing methods must be carefully tailored to specific models and only come with asymptotic theoretical guarantees. In this work, however, we propose a general model selection framework using Universal Inference, making our method widely applicable to various testing scenarios. Since Universal Inference requires two sets of data, we employ edge sampling to obtain proper networks with tractable dependence. We prove that the proposed statistic is an e-value, thus controlling the type I error rate in finite samples under nearly any hypothesis test. To our knowledge, this is the first Universal Inference-type statistic constructed from dependent splits of data as well as the first finite-sample testing guarantee for hypothesis testing on networks. We also prove that the logarithm of the test statistic diverges to positive infinity under various alternative models. On simulated and real-world networks, the proposed method performs well on tasks such as choosing the random graph model and the number of communities.

\end{abstract}

\section{Introduction}\label{sec:intro}

Model selection and hypothesis testing are classical statistical inference tasks and have important applications in network-valued data. Consider the following motivating examples. First, numerous random graph models can be fit to an observed network. A non-exhaustive list includes: Erdos-Renyi \citep{Erdos1959}, Chung-Lu \citep{chung2002average}, stochastic block model \citep{holland1983stochastic}, degree-corrected block model \citep{karrer2011stochastic}, and popularity-adjusted block model \citep{senguptapabm}. Practitioners are thus faced with a difficult task: given a network, which model is most appropriate? For example, does the network have a statistically significant community structure, or is the Erdos-Renyi model appropriate? Compared with a stochastic block model, is the extra flexibility of the degree-corrected block model truly warranted? Each of these model selection tasks can be recast as a hypothesis testing problem. For example, we may be interested in testing
$$
    H_0:\text{stochastic block model vs. }H_1:\text{degree-corrected block model,}
$$
which is effectively testing if the degree parameters in the degree-corrected block model are all identical ($H_0$) or not ($H_1$). Once a random graph model has been chosen, most models still require the user to set other parameter values. For example, in block models, the number of communities, $K$, must be pre-specified. While there may be situations where $K$ is known based on domain-knowledge, in general, it is unknown. Thus, choosing the number of communities is another important model selection task that can also be framed as a hypothesis test, i.e., 
$$
    H_0:K=2\text{ vs. }H_1:K=3.
$$

Many existing works tackle only one of these tasks. For example, for selecting the number of communities $K$, \cite{wang2017likelihood} and \cite{hu2020corrected} propose likelihood-based approaches while \cite{bickel2016hypothesis} use the largest eigenvalue to recursively test against an Erdos-Renyi null model. For selecting a random graph model, \cite{yanchenko2024generalized} propose a model-agnostic parameter with analytic and bootstrap cut-off. For these methods, however, the solution is tailored to a specific test and cannot be easily generalized to other settings. Recently, more general model selection paradigms have been introduced as a starting point and/or when no other methods exist. \cite{li2020network} propose an edge cross validation approach by randomly sampling node pairs and using matrix completion. \cite{chakrabarty2025network} and \cite{bhadra2025unified} leverage loss functions with the former sub-sampling the network and the latter proposing a unified approach to select the best block model from a natural hierarchy. While these methods are more generally applicable across various testing settings, the theoretical results must be derived for each particular model and only apply in the asymptotic regime. Moreover, the loss-function-based methods \citep[e.g.,][]{li2020network, chakrabarty2025network, bhadra2025unified} simply provide a rejection decision rather than a statistical measure of evidence such as a p- or e-value.

One recent approach to handle difficult hypothesis testing settings is Universal Inference \citep{wasserman2020universal, dey2025generalized}. Given independent copies of the data, Universal Inference uses sample splitting to construct a likelihood ratio-like statistic with finite-sample control of type I errors for nearly any sets of hypotheses. Universal Inference is also a special case of e-values, an exciting alternative to p-values \citep{vovk2021values, wang2022false, grunwald2024safe}. Given the complexities of testing on networks, Universal Inference would appear to be a promising approach. The main problem, however, is that Universal Inference requires {\it independent copies} of the data. With network data, however, we typically only obtain a single observation, and even in situations where there are multiple copies, i.e., temporal networks \citep{holme2012temporal}, the successive network realizations are likely correlated. Prior approaches to split networks would also run into difficulties if they were to be used in Universal Inference. For example, \cite{li2020network} split the network via node-pairs, but this does not result in proper matrices, thus requiring an additional matrix completion step. Data Fission \citep{leiner2025data} has also been used to split networks, but the resulting pieces exhibit a complicated dependence structure \citep{leiner2024graph, ancell2025post}.

To overcome this challenge, we propose using {\it edge sampling}, randomly dividing the observed edges between two new networks. While the resulting networks exhibit some dependence, it is analytically manageable and still allows us to construct a Universal Inference testing statistic. We prove that our statistic is an e-value under nearly any null hypothesis, meaning that our test maintains a proper type I error rate even in finite samples. To our knowledge, this is the first Universal Inference-like statistic constructed from dependent data, as well as the first finite-sample guarantee for hypothesis testing on networks. We also prove that the statistic is theoretically guaranteed to diverge to infinity under various alternative models. See Figure \ref{fig:diagram} for an overview of the proposed method. 

\begin{figure}
    \centering
    \includegraphics[width=0.99\linewidth]{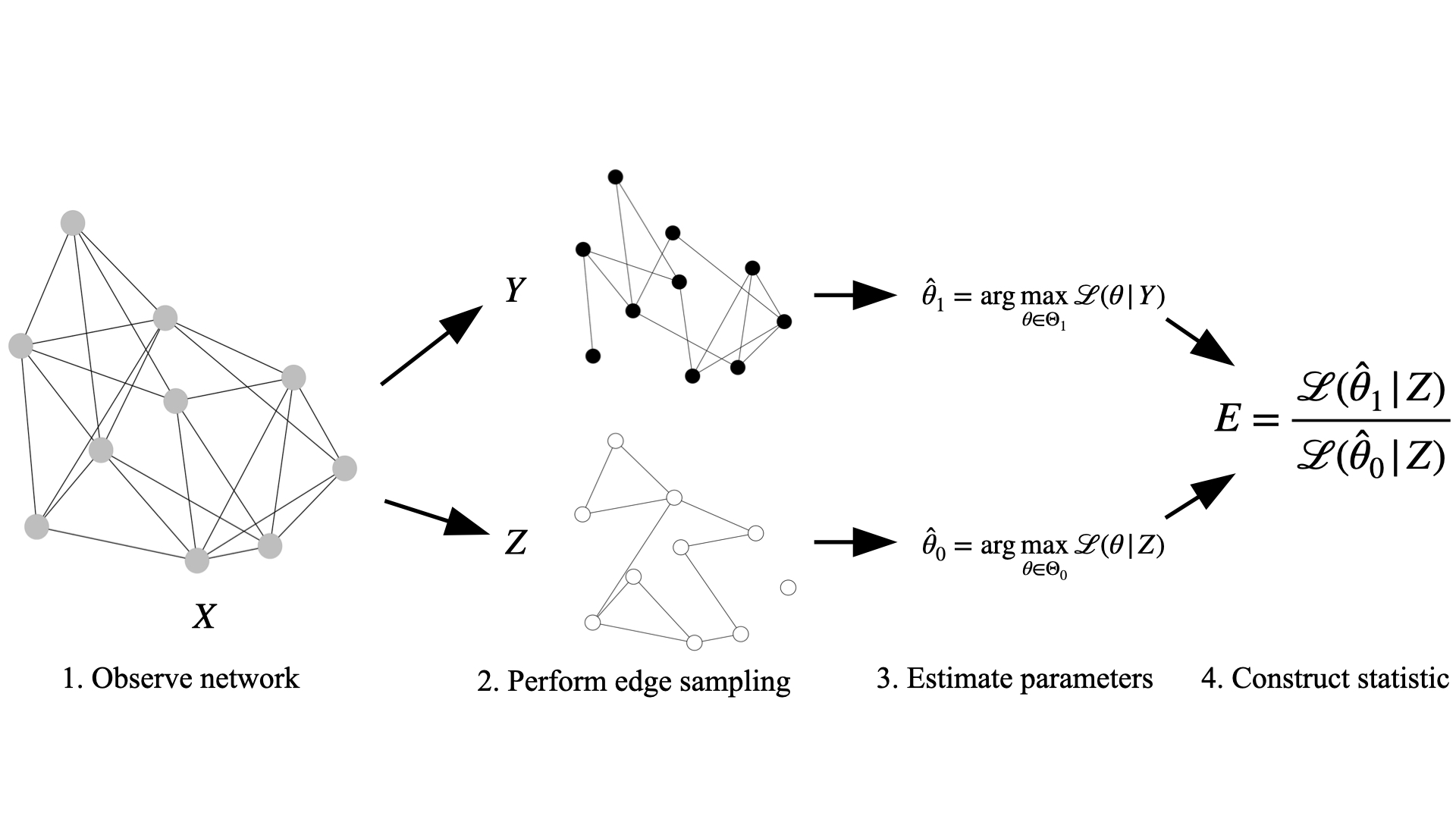}
    \caption{Visual schematic of the proposed method.}
    \label{fig:diagram}
\end{figure}

The remainder of this paper is structured as follows. In Section \ref{sec:back}, we give the necessary background on random graph models and Universal Inference. The proposed method is presented in Section \ref{sec:method} with theoretical results in Section \ref{sec:theory}. Experiments on synthetic and real-world data are the topic of Section \ref{sec:data}, and the conclusion can be found in Section \ref{sec:conc}.

\section{Background}\label{sec:back}
\subsection{Notation}
We begin by introducing notation before providing the necessary background for this work. We represent a network with $n$ nodes by an $n\times n$ adjacency matrix $\bA\in\{0,1\}^{n\times n}$ where $A_{ij}=1$ if nodes $i$ and $j$ have an edge, and 0 otherwise. For simplicity, we only consider undirected, unweighted networks with no self-loops, but the ideas in this work can easily be generalized. We let $\bP(\btheta)$ represent the $n\times n$ data-generating matrix defined by parameters $\btheta$ such that $\bA\sim \bP(\btheta)$ corresponds to
$$
    A_{ij}\mid\btheta\stackrel{\text{ind.}}{\sim}\mathsf{Bernoulli}(P_{ij}(\btheta))
$$
for all $1\leq i<j\leq n$. When we want to emphasize that the network was generated from the {\it true} model parameters, we will denote this with a star, i.e., $\btheta^*$. 

When describing asymptotic rates, we say that a deterministic sequence $a_n=\mathcal O(b_n)$ if $\lim_{n\to\infty}a_n/b_n= c$ for some constant $c$, and $a_n=o(b_n)$ if $\lim_{n\to\infty}a_n/b_n=0$. For a sequence of random variables, $X_n=\mathcal O_p(a_n)$ if for any $\epsilon>0$ and some finite $N,M>0$, $\mathbb P(|X_n/a_n|>M)<\epsilon$ for all $n>N$; and $X_n= o_p(a_n)$ if for any $\epsilon>0$, $\lim_{n\to\infty}\mathbb P(|X_n/a_n|>\epsilon)=0$. We also use $\sim$ as a shorthand whenever the result is $\mathcal O(\cdot)$ or $\mathcal O_p(\cdot)$. Additionally, $X_n\lesssim Y_n$ means $X_n=\mathcal O_p(Y_n)$ and $X_n\leq Y_n$ for all $n$.

\subsection{Random graph models}\label{sec:rgm}
We now discuss the form of $\bP(\btheta)$ for several popular random graph models. First, the stochastic block model (SBM) \citep{holland1983stochastic} is a popular model of meso-scale structures in networks. We assume that nodes are partitioned into $K$ groups such that $\bc\in\{1,\dots,K\}^n$ is the group assignment vector where $c_i=k$ means that node $i$ belongs to group $k$ for $i\in\{1,\dots,n\}$ and $k\in\{1,\dots,K\}$. Then the probability of two nodes forming an edge depends solely on their group membership through the $K\times K$ matrix ${\bf B}$. Specifically, $\btheta=\{{\bf B},\bc\}$ and $P_{ij}(\btheta)=B_{c_i,c_j}$. A special case of the SBM is the Erdos-Renyi (ER) model \citep{Erdos1959} where $K=1$ and $P_{ij}(\btheta)=p$ for all $(i,j)$. In other words, the probability of an edge is the same for all node pairs.

As an extension of the SBM, the degree-corrected block model (DCBM) \citep{karrer2011stochastic} allows for degree heterogeneity. In addition to the group labels $\bc$ and probability matrix $\bB$, this model introduces a node-specific weight parameter $\bpsi\in(0,1)^n$ where $\psi_i$ corresponds to the degree of node $i$, i.e., larger $\psi_i$ means larger expected degree. In this way, the DCBM can model both block structure (through $\bB$) and degree heterogeneity (through $\bpsi$) in a network. Then the probability of an edge between nodes is $P_{ij}(\btheta)=\psi_iB_{c_i,c_j}\psi_j$ where $\btheta=\{\bc,\bB,\bpsi\}$. If the $B_{jk}$'s are equal for all $j$ and $k$, then the DCBM reduces to the Chung-Lu (CL) model \citep{chung2002average}, i.e., $P_{ij}(\btheta)=\psi_i\psi_j$, and does not encode block structure.

Finally, as an extension to the DCBM, \cite{senguptapabm} propose the popularity-adjusted block model (PABM) which allows each node to have a different popularity parameter for each group. Specifically, the authors define $\blambda\in(0,1)^{n\times K}$ where $\lambda_{ik}$ corresponds to the propensity of node $i$ to form an edge with a node in group $k$. Then $\btheta=\{\bc,\blambda\}$ and $P_{ij}(\btheta)=\lambda_{i,c_j}\lambda_{j,c_i}$. For identifiability, we will always assume that $\Lambda_{k\ell}=\Lambda_{\ell k}$ for all $k,\ell\in\{1,\dots,K\}$ where $\Lambda_{k\ell}=\sum_{i=1}^n\lambda_{i\ell}\mathds {1}(c_i=k)$ and $\mathds{1}(\cdot)$ is the indicator function. In words, this condition requires that the sum of the parameters for nodes in group $k$ connecting with group $\ell$ must equal the sum of the parameters for nodes in group $\ell$ connecting with group $k$. Clearly, the ER, SBM, CL and DCBM are all special cases of the PABM.

\subsection{Universal Inference}
Universal Inference \citep{wasserman2020universal, dey2025generalized} was proposed as a general hypothesis testing framework with finite-sample type I error guarantees. Let $X_1,\dots,X_{2n}\stackrel{\text{iid.}}{\sim} P_{\btheta^*}$ where $\btheta^*$ is the true parameter value which generated the data and $P_{\btheta^*}\in \{P_{\theta}:\btheta\in\Theta\}$. Additionally, let $\mathcal L(\btheta|\bX)$ be the associated likelihood function, but mis-specified models \citep{park2025robust} and/or loss functions based on risk-minimizers also exist \citep{dey2025generalized}.

Suppose the goal is to test $H_0:\btheta^*\in\Theta_0$ vs.~$H_1:\btheta^*\in\Theta_1$. In a typical hypothesis testing setting, we would derive the distribution of some test statistic $T(\btheta)$ when $\btheta\in\Theta_0$. As discussed in Section \ref{sec:intro}, however, this can be quite difficult with network data due to the inherent correlation in the data. 

Instead, the Universal Inference procedure splits the data into two parts, $D_0$ and $D_1$. We calculate $\hat\btheta_1\in\Theta_1$ which is any estimator under the alternative hypothesis constructed using {\it only} the data from $D_1$. We then evaluate the likelihood function using the data from $D_0$ but the parameters estimated from $D_1$, i.e., $ L_1=\mathcal L(\hat\btheta_1|D_0)$. Next, we estimate the parameter value $\hat\btheta_0\in\Theta_0$ under the null hypothesis using only the data from $D_0$, and then evaluate the likelihood using this same data, i.e., $ L_0=\mathcal L(\hat\btheta_0|D_0)$. We stress that $ L_1$ is evaluated on $D_0$ but the parameter estimates come from $D_1$; on the other hand, $L_0$ is both evaluated and estimated from $D_0$. The ratio of the two likelihoods is then the test statistic, i.e., 
\begin{equation}\label{eq:eval}
    E = \frac{ L_1}{ L_0}.
\end{equation}
\cite{wasserman2020universal} prove that if we reject $H_0$ whenever $E>1/\alpha$, then this test has a type I error rate of at most $\alpha$. This result holds for almost any testing scenario and does not require any further derivations. The statistic in \eqref{eq:eval} is also a special case of an e-value, a testing alternative to p-values \citep[e.g.,][]{grunwald2024safe}. A statistic $E$ is an e-value if $\mathbb E_{\btheta}E\leq1$ for all $\btheta\in\Theta_0$, i.e., the expected value is at most one under the null hypothesis. 

\section{Methodology}\label{sec:method}

\subsection{Procedure}\label{sec:steps}
In this section, we propose the main method of this paper.
Assume that we observe a network $\bA\sim \bP(\btheta)$ and are interested in testing
$$
    H_0:\btheta\in\Theta_0\text{ vs. }H_1:\btheta\in\Theta_1,
$$
where $\Theta_0$ and $\Theta_1$ correspond to different models. For example, $\Theta_0$ may correspond to an ER model while $\Theta_1$ might be a 2-block SBM. Regardless, we can write down a likelihood for each model, $\mathcal L(\bP(\btheta)|\bA)$. Given these hypotheses, we propose the following model-selection framework using Universal Inference.

\paragraph{Step 1: Edge sampling.}
To utilize Universal Inference, we first must split the data into two parts, but this is challenging with network data where we typically only obtain a single observation. One approach is sampling {\it node pairs} as in \cite{li2020network} since these are assumed to be conditionally independent. Sampling node pairs, however, does not lead to proper matrices (networks) so the authors must perform an additional matrix completion step to obtain a proper matrix (network). Another possibility to split the network would be to use Data Fission \citep[e.g.,][]{leiner2024graph, leiner2025data, ancell2025post}, but since the data are constructed from Bernoulli data, the resulting data splits exhibit a complicated dependence.

We instead propose using {\it edge sampling} to generate two copies of data from the observed network. Specifically, we construct two networks $\bY$ and $\bZ$ defined by 
$$
    Z_{ij}\sim
    \begin{cases}
        0&\text{if }A_{ij}=0\\
        \mathsf{Bernoulli}(\gamma)&\text{if }A_{ij}=1
    \end{cases}
$$
and $\bY=\bA-\bZ$ where $\gamma\in(0,1)$ is the probability that an edge from $\bA$ is allocated to $\bZ$. In words, if $A_{ij}=1$, then this edge will appear in $\bZ$ with probability $\gamma$, and in $\bY$ with probability $1-\gamma$. Marginally, $Y_{ij}\stackrel{\text{ind.}}{\sim}\mathsf{Bernoulli}((1-\gamma)P_{ij}(\btheta))$ and $Z_{ij}\stackrel{\text{ind.}}{\sim}\mathsf{Bernoulli}(\gamma P_{ij}(\btheta))$, but $Z_{ij}\not\perp Y_{ij}$. We will need to carefully account for this fact when deriving theoretical results. Please see Section \ref{supp:joint} of the Supplemental Materials for a more detailed discussion of the joint distribution between $Y_{ij}$ and $Z_{ij}$. Notably, we can derive the joint distribution exactly, and if the network is sparse, then for large $n$, $Y_{ij}$ and $Z_{ij}$ are approximately independent. Regardless, both $\bY$ and $\bZ$ yield information on the parameters of interest $\bP(\btheta)$. Additionally, if $\gamma$ is large, then more edges will be allocated to $\bZ$ than $\bY$ (with high probability), with $\gamma=1$ resulting in $\bZ=\bA$ and $\bY={\bf 0}_{n\times n}$.

\paragraph{Step 2: Find parameters under $H_1$ and evaluate on $\bZ$.}
In this step, we use $\bY$ to estimate $\btheta$ under the alternative hypothesis, i.e., $\hat\btheta_1=\arg\max_{\btheta\in\Theta_1}\mathcal L(\bP(\btheta)|\bY)$. Then using $\hat\btheta_1$, we evaluate the likelihood on $\bZ$, i.e., $\mathcal L(\frac{\gamma}{1-\gamma}\bP(\hat\btheta_1)|\bZ)$ where
$$
    \mathcal L(\bP(\btheta)\mid \bA)
    =\prod_{i<j}e^{-P_{ij}(\btheta)}P_{ij}(\btheta)^{A_{ij}}.
$$
We use the Poisson likelihood as an approximation to the Bernoulli likelihood to simplify theoretical derivations, which is a standard approximation in the literature \citep[e.g.,][]{senguptapabm}.  Additionally, the ratio of $\gamma$ terms is motivated by the edge sampling step. Recall that $\bY$ is such that $Y_{ij}\sim\mathsf{Bernoulli}((1-\gamma)P_{ij}(\btheta))$. Therefore, when we find $\hat\btheta_1$, this really yields the MLE for $(1-\gamma)\bP(\btheta)$. Thus, we divide by $1-\gamma$ to yield the MLE for $\bP(\btheta)$. But since this will be applied to the likelihood of $\bZ$ which has marginal distribution $Z_{ij}\sim\mathsf{Bernoulli}(\gamma P_{ij}(\btheta))$, multiplying by $\gamma$ will better match the likelihood of $\bZ$. For this reason, we calculate the likelihood as $L_1=\mathcal L(\frac{\gamma}{1-\gamma}\bP(\hat\btheta_1)|\bZ)$ with the correction for $\gamma$.

\paragraph{Step 3: Find parameters under $H_0$ and evaluate on $\bZ$.} 
Using $\bZ$ and assuming that the null hypothesis is true, we find $\hat\btheta_0=\arg\max_{\btheta\in\Theta_0}\mathcal L(\bP(\btheta)|\bZ)$. The likelihood function is then evaluated on $\bZ$, i.e., $L_0=\mathcal L(\bP(\hat\btheta_0)|\bZ)$. Since the parameters were estimated from the same data that are being used to evaluate the likelihood, we do not require any factors of $\gamma$. Notice this differs from step 2 where the parameter estimates and likelihood evaluation come from different networks. 

\paragraph{Step 4: Calculate e-value and reject if large.}
Finally, we compute the test statistic as the ratio of the two likelihoods, $E=L_1/L_0$, and reject if $E$ is large. We prove below that $E$ is an e-value. One complication is that we constructed the statistic using a Poisson likelihood, but the data are truly coming from a Bernoulli model; i.e., mis-specified model \citep[e.g.,][]{park2025robust}. Even more importantly, however, we do not have independence between $\bY$ and $\bZ$, which is usually needed when constructing the Universal Inference e-value statistic.  To our knowledge, this is the first Universal Inference e-value constructed from dependent splits of data.

\subsection{Example: Erdos-Renyi vs.~Stochastic Block Model}
To fix ideas, we detail the procedure for the specific example of testing if the network was generated from $H_0$: Erdos-Renyi vs.~$H_1$: SBM with $K=2$ communities. For the SBM, the likelihood function (assuming the Poisson approximation) is
$$
    \mathcal L_{SBM}(\bP(\btheta)|\bA)
    =\prod_{i<j} e^{-B_{c_i,c_j}} B_{c_i,c_j}^{A_{ij}}
$$
where $\btheta=\{ {\bf B},\bc\}$, ${\bf B}\in(0,1)^{K\times K}$ is the block connectivity matrix, and $\bc\in\{1,\dots,K\}^n$ is the community assignment vector, i.e., $c_i=k$ if node $i$ is in community $k$. Notice that the Erdos-Renyi model is a special case of the SBM, where $K=1$:
$$
    \mathcal L_{ER}(\bP(\btheta)|\bA)
    =\prod_{i<j}e^{-p} p^{A_{ij}}
    =e^{-pn(n-1)/2}p^{\sum_{i<j}A_{ij}}, \quad \btheta=\{p\}.
$$
Assume we observe a network $\bA$ and then use edge sampling to split it in two, $\bY$ and $\bZ$. Using $\bY$, we estimate the parameters of a two-block SBM, namely the community labels and block connectivity parameters, i.e.,
$$
    \hat\btheta_1=\{\hat{\bf B}, \hat\bc\}
    =\arg\max_{{\bf B},\bc} \prod_{i<j} e^{-B_{c_i,c_j}} B_{c_i,c_j}^{Y_{ij}}.
$$
Since $Z_{ij}\sim\mathsf{Bernoulli}(\gamma P_{ij}(\btheta))$ and $Y_{ij}\sim\mathsf{Bernoulli}((1-\gamma)P_{ij}(\btheta))$, we need to scale the estimates $\hat{\bf B}$ by $\gamma/(1-\gamma)$ to properly evaluate the likelihood function on $\bZ$. We then use these parameter estimates to evaluate the SBM likelihood on $\bZ$, i.e.,
$$
    L_1
    =\mathcal L_{SBM}(\tfrac{\gamma}{1-\gamma}\bP(\hat\btheta_1)|\bZ)
    =\prod_{i<j} e^{-\frac{\gamma}{1-\gamma}\hat B_{\hat c_i,\hat c_j}} \left(\tfrac{\gamma}{1-\gamma}\hat B_{\hat c_i,\hat c_j}\right)^{Z_{ij}}.
$$
Next, we estimate the Erdos-Renyi model parameter using $\bZ$
$$
    \hat p
    =\frac1{n(n-1)/2}\sum_{i<j} Z_{ij},
$$
and evaluate
$$
  L_0
  =\mathcal L_{ER}(\bP(\hat p)|\bZ)
  =e^{-\hat pn(n-1)/2}\hat p^{\sum_{i<j} Z_{ij}}.
$$
Finally, we compute $E=L_1/L_0$ and reject the null hypothesis if $E$ is large.


\subsection{Hyper-parameter selection}
While the approach is relatively automatic, there are a few considerations practitioners must be aware of when implementing it. Of course the models, number of communities, etc.~for the null and alternative hypotheses must be carefully chosen based on domain knowledge and the network of interest. Moreover, since the edge sampling step introduces some randomness, we suggest running the algorithm multiple times and taking the average e-value. Unless otherwise specified, we report the mean e-value over 100 data splits.

The other hyper-parameter that must be carefully selected is $\gamma\in(0,1)$ which controls the ``amount'' of data contained in $\bY$ and $\bZ$. If $\gamma$ is closer to zero, then $\bY$ will be relatively dense, while $\bZ$ will have comparatively more edges if $\gamma$ is near one. Naively, we may think that setting $\gamma=0.5$ is the best choice to ensure a similar density of $\bY$ and $\bZ$, but this can be shown intuitively and empirically not to be the case. Consider a simple example of testing $H_0$: ER vs.~$H_1$: SBM with $K=2$ communities. Based on our procedure, we split the adjacency matrix $\bA$ into $\bY$ and $\bZ$. Then $\bY$ is used to estimate the model parameters under the alternative hypothesis, which here means the community labels $\bc$ and block probabilities $\bB$, i.e., $n+K(K-1)/2$ parameters. On the other hand, the model parameters for the null model are estimated from $\bZ$, which, in this case, simply means estimating one parameter $p$, the mean edge probability of the ER model. Clearly, estimation under the null model is far easier and does not require as much data in order to achieve a good estimate. The alternative model, conversely, has many more parameters to estimate. Therefore, we should set $\gamma$ close to 0 to ensure that $\bY$, which is used to estimate the more complicated alternative model, has more edges. We described this for the case of testing a stochastic block model, but the idea holds more generally as typically the null model will be simpler compared to the alternative. Thus, we suggest setting $\gamma<0.5$ to improve the power of the test, with simpler null models allowing for even smaller values, i.e., $\gamma=0.1$. This approach is validated in the empirical experiments and aligns with suggestions in \cite{li2020network}.

\section{Theoretical results}\label{sec:theory}
In this section, we theoretically study the performance of the proposed hypothesis test, with all proofs left to the Supplemental Materials, Sections \ref{supp:validity}-\ref{supp:sbm}.
We begin by proving that the proposed method yields an e-value, which guarantees a type I error rate of at most $\alpha$, only requiring a single, minor assumption:
\begin{assumption}\label{a:edge}
    {\it Let $\bP(\hat\btheta)=\arg\max_{\btheta\in\Theta}\mathcal L(\bP(\btheta)|\bA)$. Then}
    $
    \sum_{ij} P_{ij}(\hat\btheta)=\sum_{ij}A_{ij}.
    $
\end{assumption}

\noindent
\cite{senguptapabm} show that \ref{a:edge} holds for the PABM, of which the ER, CL, SBM, and DCBM are all special cases. It is not clear that this necessarily holds for {\it all} models, e.g., Random Dot Product Graphs, which is why we explicitly include the assumption.

\begin{theorem}\label{thm:validity}
{\it Let $\bA\sim\bP(\btheta^*)$ and consider testing}
$$
    H_0:\btheta^*\in\Theta_0\text{ vs. }H_1:\btheta^*\in\Theta_1.
$$
{\it Under \ref{a:edge} and for any $n$, the statistic constructed in Section \ref{sec:steps} is an e-value, i.e.,}
$$
    \sup_{\btheta\in\Theta_0}\mathbb E_{\btheta}E\leq 1.
$$  
\end{theorem}

\noindent
While the proof technique was inspired by \cite{wasserman2020universal}, there are various challenges here that require extra care. First, our e-value statistic is constructed assuming a Poisson model, but the edges are actually generated from a Bernoulli distribution. This bears some similarities to Universal Inference for mis-specified models \citep{park2025robust}. More important, however, is handling the dependence between $\bY$ and $\bZ$. To do this, we took an iterated expectation and leveraged the conditional distribution of $Z_{ij}\mid Y_{ij}$ (Section \ref{supp:joint}). After taking the expectation conditional on $\bY$, we prove that the resulting term is deterministically less than one starting from the definition of $\bP(\hat\btheta)$. To our knowledge, this is the first Universal Inference-type statistic constructed from dependent splits of data.

\begin{corollary}\label{cor:error}
  {\it Under \ref{a:edge}, the testing procedure described in Section \ref{sec:steps} has a type I error rate of at most $\alpha$, i.e.,}
$$
    \sup_{\btheta\in\Theta_0}\mathbb P_{\btheta}(E>1/\alpha)\leq \alpha.
$$
\end{corollary}

\noindent
The corollary follows immediately from Theorem \ref{thm:validity} by applying Markov's inequality. We stress that this result applies for finite samples, whereas existing network hypothesis testing methods only establish type I error rates in an asymptotic regime \citep[e.g.,][]{bickel2016hypothesis, wang2017likelihood, yanchenko2024generalized, chakrabarty2025network, ancell2025post}. We also emphasize the generality of this theorem; we derived a single result that applies to nearly any random graph model of the form in Section \ref{sec:rgm}. For example, testing the number of communities in a network or which network model is most appropriate are both covered by this theorem and ensure that the proposed approach maintains a proper type I error. This again contrasts with previous works, which need to prove separate results for each null model. Moreover, previous general model selection approaches \citep[e.g.,][]{li2020network, chakrabarty2025network} only yield a model selection decision, i.e., which model fits best, but do not provide a statistical measure of evidence like a p- or e-value.

While we have shown that our statistic is an e-value, simply setting $E=1$ independent of the network would also yield an e-value, even though this is clearly a meaningless statistic. Therefore, we want to show that our proposed hypothesis test yields large values of $E$ when the network is generated under the alternative model. In the following theorem, we prove this for testing an ER model against a PABM (given the following assumptions) where we emphasize that the following are asymptotic results by using subscripts on $\bP_n$, $\bA_n$ and $E_n$.
\begin{assumption}\label{a:sparse}
    $\bP_n(\btheta^*)=\varrho_n\bP(\btheta^*)$ where $\varrho_n=o(1)$ and $n\varrho_n^2/\log^2n\to\infty$
\end{assumption}

\begin{assumption}\label{a:K}
   $K$ is fixed and known and the number of nodes in each (true and estimated) community is $\mathcal O(n)$ 
\end{assumption}

\begin{assumption}\label{a:pmin}
    $\min_{(i,j)}\{P_{ij}(\btheta^*)\}>0$
\end{assumption}

\begin{assumption}\label{a:ident}
    For any two nodes $j_1,j_2$ where $c^*_{j_1}\neq c_{j_2}^*$, the set $\{P_{ij_1}(\btheta^*)/P_{ij_2}(\btheta^*)\}_{i=1}^n$ has at least $K+1$ distinct values 
\end{assumption}

\begin{assumption}\label{a:set}
    If the average interaction level between any two large, distinct subsets of the vertex set is non-zero, it must be at least $\mathcal O(\varrho_n/\log n)$
\end{assumption}

\begin{assumption}\label{a:pbar}
    $\bar p=n^{-2}\sum_{ij}P_{ij}(\btheta^*)$ does not depend on $n$
\end{assumption}

\begin{theorem}\label{thm:er}
    {\it Let $\bA_n\sim \varrho_n\bP(\btheta^*)$ under \ref{a:sparse}. Assume that we are testing 
$$
    H_0:\btheta^*\in\Theta_{0}\text{ vs. }\btheta^*\in\Theta_1
$$
where $\Theta_{0}$ is the parameter space corresponding to the Erdos-Renyi model, i.e., $P_{ij}(\btheta^*)=p$ for all $(i,j)$, and $\Theta_1$ corresponds to the popularity-adjusted Block Model, i.e., $P_{ij}(\btheta^*)=\lambda_{ic_j^*}^*\lambda_{jc_i^*}^*$, under \ref{a:K}--\ref{a:pbar}. If $\btheta^*\in\Theta_1$ such that $\bA_n$ is generated from the alternative model, then there exists some $\epsilon>0$ such that, as $n\to\infty$, }
$$
    \mathbb P(\log E_n>\epsilon\varrho_nn^2)\to 1.
$$    
\end{theorem}

\noindent
This result implies that, under the alternative hypothesis and with high probability, the e-value grows at a rate of $\exp(\varrho_nn^2)$, diverging to positive infinity. Moreover, the asymptotic power of the test --- that rejects the null hypothesis whenever $E_n>1/\alpha$ --- is 1 for any $\alpha\in(0,1)$. Additionally, assumptions \ref{a:sparse}, \ref{a:K}, \ref{a:ident} and \ref{a:set} are technical assumptions inherited from \cite{senguptapabm} to establish that the estimated community labels of the PABM, $\hat \bc$, are asymptotically close to the true data-generating labels, $\bc^*$, and that the model parameter estimates, $\hat\lambda_{ik}$, converge to the true values, $\lambda^*_{ik}$. We highlight that \ref{a:sparse} means the networks are sparse since $\varrho_n=o(1)$, but the result also holds for dense networks, i.e., $\varrho_n=O(1)$. Additionally, the communities in the PABM are detectable thanks to \ref{a:ident} while \ref{a:pmin} ensures that $\log P_{ij}(\btheta^*)$ is well-defined for all $(i,j)$; an analogous assumption is made in \cite{chakrabarty2025network}.

To our knowledge, this is the first network model selection result incorporating the PABM as an alternative model. The CL, SBM and DCBM models are all special cases of the PABM so Theorem \ref{thm:er} can easily be adapted to each of these models. That being said, modifying the proof would require adjusting \ref{a:ident} for these special cases. In particular, for the CL model, \ref{a:ident} simplifies to requiring that not all weight parameters are equal such that the CL model is truly distinct from an ER model. For the SBM, \ref{a:ident} changes to requiring that any two rows in the block probability matrix $\bB$ have at least one disagreement, a standard assumption in the SBM literature \cite[e.g.,][]{bickel2009nonparametric}. For the DCBM, \ref{a:ident} would reduce to requiring $\{P_{ij_1}(\btheta^*)/P_{ij_2}(\btheta^*)\}_{i=1}^n$ to take at least two distinct values when $c_{j_1}^*\neq c_{j_2}^*$.

In the e-value literature, it is common to show that the expected value of the logarithm of the statistic diverges to positive infinity, i.e., $\mathbb E(\log E_n)\to\infty$ as $n\to\infty$ \citep[e.g.,][]{grunwald2024safe}. Due to the discrete nature of the network data, however, there is a small, but non-zero probability that $E_n=0$, which means $\mathbb E(\log E_n)=-\infty$. We formalize this in the following proposition for Bernoulli data, but stress that it more generally applies for random variables with non-zero probability of being zero.

\begin{proposition}\label{prop:ui}
    {\it Let $X_1,\dots,X_{2n}\stackrel{\text{iid.}}{\sim}\mathsf{Bernoulli}(\theta)$ and $E_n$ be the split-likelihood Universal Inference statistic from \cite{wasserman2020universal}. Assume that we are testing}
$$
    H_0:\theta=\theta_0\text{ vs. }\theta\neq \theta_0,
$$
{\it where $0<\theta_0<1$. Then $\mathbb E(\log E_n)=-\infty$.}
\end{proposition}

\noindent
One approach to avoid this issue would be to ``pad'' the estimate of $\hat\theta$ such that it is bounded away from zero and one, but it is not obvious how this would affect the properties of the statistic. We view this as an open problem in the e-value / Universal Inference literature that is worthy of future consideration. \\

\noindent
Next, we prove a similar result for testing against the CL null model after adding an additional assumption.

\begin{assumption}\label{a:cl}
    The vector $(\bar P_{i}\bar P_{j})_{ij}$ is linearly independent of $(P_{ij}(\btheta^*))_{ij}$ where $ \bar P_{i}(\btheta^*):=\tfrac1n\sum_{j=1}^nP_{ij}(\btheta^*)$ for all $i\in\{1,\dots,n\}$.
\end{assumption}

\begin{theorem}\label{thm:cl}
   {\it Let $\bA_n\sim \varrho_n\bP(\btheta^*)$ under \ref{a:sparse}. Assume that we are testing 
$$
    H_0:\btheta^*\in\Theta_{0}\text{ vs. }\btheta^*\in\Theta_1
$$
where $\Theta_{0}$ is the parameter space corresponding to the Chung-Lu model, i.e., $P_{ij}(\btheta^*)=\psi_i\psi_j$, and $\Theta_1$ corresponds to the popularity-adjusted Block Model, i.e., $P_{ij}(\btheta^*)=\lambda_{ic_j^*}^*\lambda_{jc_i^*}^*$, under \ref{a:K}--\ref{a:cl}. If $\btheta^*\in\Theta_1$ such that $\bA_n$ is generated from the alternative model, then there exists some $\epsilon>0$ such that, as $n\to\infty$, }
$$
    \mathbb P(\log E_n>\epsilon\varrho_nn^2)\to 1.
$$ 
\end{theorem}

\noindent
The proof follows similarly to that of Theorem \ref{thm:er} with extra care needed for the additional parameters in the CL null model. The extra assumption ensures that the alternative model is strictly different from the CL model. In other words, if \ref{a:cl} were not true, then the PABM could reduce to a CL model. Again, the SBM and DCBM are special cases of the PABM, so this theorem can be adapted to these alternative models as well.\\

\noindent
Finally, we prove an analogous result for testing for the number of communities in an SBM. 

\begin{theorem}\label{thm:sbm}
    {\it Let $\bA_n\sim \varrho_n\bP(\btheta^*)$ under \ref{a:sparse}. Assume that we are testing 
$$
    H_0:\btheta^*\in\Theta_{0}\text{ vs. }\btheta^*\in\Theta_1
$$
where $\Theta_{0}$ is the parameter space corresponding to the stochastic block model with $K-1>1$ blocks, and $\Theta_1$ corresponds to the stochastic block model with $K$ blocks under \ref{a:K}, \ref{a:pmin} and \ref{a:pbar} where the block probability matrix is also of full rank. If $\btheta^*\in\Theta_1$ such that $\bA_n$ is generated from the alternative model, then there exists some $\epsilon>0$ such that, as $n\to\infty$, }
$$
    \mathbb P(\log E_n>\epsilon\varrho_nn^2)\to 1.
$$
\end{theorem}

\noindent
We sketch a proof in Section \ref{supp:sbm} by leaning on the results from \cite{wang2017likelihood} and  \cite{hu2020corrected}. These works use the Bernoulli likelihood as opposed to the Poisson approximation which is why we refer to this as a proof sketch. The same proof technique could also be applied to testing for the number of communities in a DCBM. 

Ideally, we would also like to derive the asymptotic growth rate of our proposed statistic for testing, e.g., $H_0:$ SBM with $K$ blocks vs.~$H_1:$ DCBM with $K$ blocks. This is difficult, however, as we need to understand the behavior of the community labels estimated assuming an SBM when in reality the network is generated from a DCBM. Thus, similar to \cite{li2020network}, we omit this theoretical result but demonstrate it numerically in the following section.

\section{Numerical experiments}\label{sec:data}
\subsection{Simulation studies}
In this section, we validate the empirical performance of the proposed method on simulated data. The goal is to show that our method yields small values of $E$ when the network is generated from the null model (low type I error rate) and larger values under the alternative (high power). For each setting, we generate a network and record whether the null hypothesis is rejected or not with results averaged over 200 Monte Carlo iterations. For each network, we run the proposed algorithm 100 times, i.e., perform edge sampling 100 times, and report the average e-value. We reject the null hypothesis if $E>20$, corresponding to a type I error rate of $0.05$. We also vary $\gamma$, which controls the amount of data in each part of the split. The proposed method is compared against NETCROP \citep{chakrabarty2025network} and ECV \citep{li2020network}, two popular methods for model selection. Code to implement the proposed method and replicate all results in this section is available at the author's GitHub: \url{https://github.com/eyanchenko/network_model_selection}.

\subsubsection{Number of communities in SBM}
First, we test for the number of communities in an SBM. Let $n=1000$ be the number of nodes and $K$ be the number of blocks or communities. We define the block-block probability matrix as ${\bf B}=\alpha\{\beta{\bf I}_K+(1-\beta){\bf J}_K\}$ where $\alpha$ controls the sparsity, $\beta$ controls the strength of the community structure and ${\bf I}_K$ and ${\bf J}_K$ are the identity matrix and matrix of all ones, respectively. Unless otherwise noted, $\alpha=0.05$, and is not to be confused with the type I error rate. We also define $\boldsymbol{\pi}\in(0,1)^K$ with $\pi_k$ as the probability that a node is assigned to community $k\in\{1,\dots,K\}$ and $\sum_{k=1}^K \pi_k=1$. For the proposed method, we follow the steps in Section \ref{sec:steps} where we use the spectral clustering algorithm \citep{rohe11} (implemented in the \texttt{randnet} R package) to estimate the community labels, and vary $\gamma\in\{0.1, 0.5, 0.9\}$. Estimating the community labels under the alternative hypothesis accounts for the majority of the computing time for our method.

In the first setting, we test $H_0:K=1$ vs. $H_1:K=2$. Notice that an SBM with $K=1$ is an ER model, so this is effectively testing if there is a community structure in the network. We first set $\boldsymbol{\pi}=(0.75, 0.25)^\top$ and vary $\beta\in\{0.400,0.425,\dots,0.900\}$ where larger $\beta$ corresponds to stronger community structure. Since $\beta>0$, the data are generated from the alternative hypothesis and we expect a large rejection rate. The results are in Figure \ref{fig:sbmK2} (left). The proposed method with $\gamma=0.1$ performs the best as it has the largest power for all $\beta>0.4$ and yields a rejection rate of one for $\beta>0.6$. The next best method is ECV which generally demonstrates a non-decreasing power. While the proposed method with $\gamma=0.5$ eventually has a large rejection rate, $\gamma=0.9$ and NETCROP both show weak performance.

The results in Figure \ref{fig:sbmK2} (right) correspond to the same test and settings except now we fix $\beta=0.6$ and vary $\delta\in[0.01, 0.50]$ of length 20 where $\boldsymbol{\pi}=(\delta, 1-\delta)^\top$. As $\delta$ increases, the size of the communities becomes more equal, making the community detection task easier. Similar to the previous setting, the proposed method with $\gamma=0.1$ performs the best followed by the ECV and $\gamma=0.5$ methods. Again, the e-value with $\gamma=0.9$ and NETCROP methods never reject the null hypothesis. For the last setting, we test $H_0:K=4$ vs. $H_1: K=5$ such that the null and alternative hypotheses correspond to proper SBMs. We fix $\beta=0.8$ and vary $\delta\in[0.01, 0.20]$ of length 20 where $\boldsymbol{\pi}=((1-\delta)/(K-1),\dots,(1-\delta)/(K-1),\delta)^\top$. Again, the networks with more balanced communities (larger $\delta$) should make it easier to reject the null. The results are in Figure \ref{fig:sbmK5}. Interestingly, the proposed method with $\gamma=0.1$ yields a rejection rate of one for all $\delta$ while for ECV, there is an increasing rejection rate which reaches a power of one by $\delta=0.20$. 

\begin{figure}
    \centering
    \includegraphics[width=0.9\linewidth]{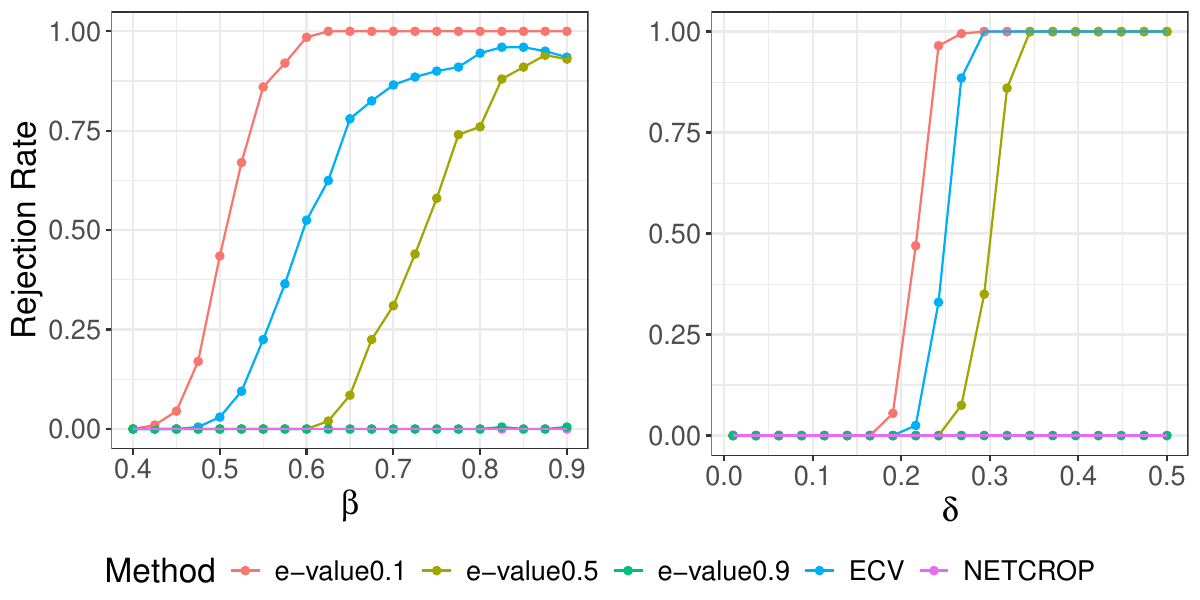}
    \caption{(left) $H_0: K=1$ (ER) vs. $H_1: K=2$ (SBM) with increasing community strength $\beta\in[0.4,0.9]$, and $\delta=0.25$. All networks are generated from the alternative model. (right) $H_0: K=1$ (ER) vs. $H_1: K=2$ (SBM) with increasing community size $\delta\in[0.01, 0.5]$, and $\beta=0.6$. All networks generated from the alternative model. }
    \label{fig:sbmK2}
\end{figure}

\begin{figure}
    \centering
    \includegraphics[width=0.55\linewidth]{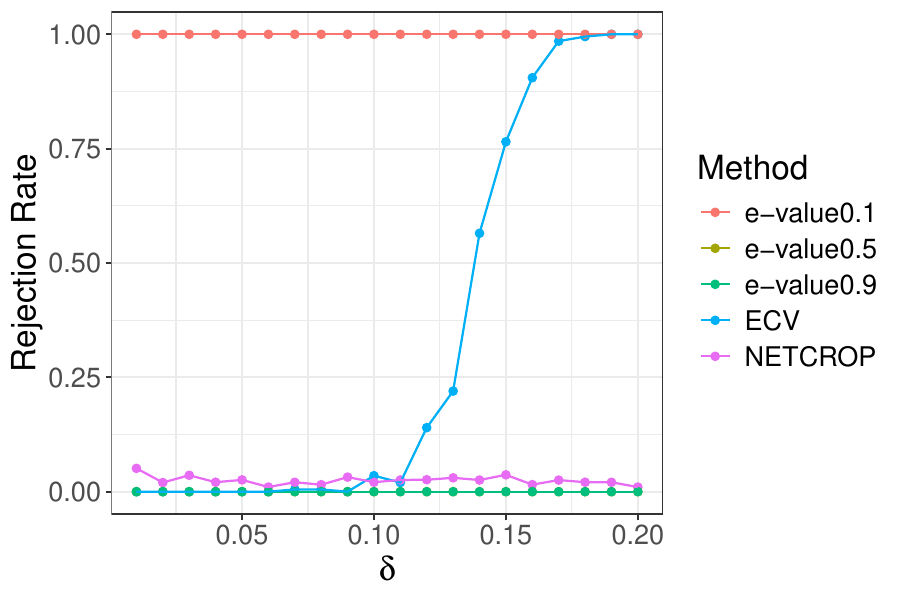}
    \caption{$H_0: K=4$ (SBM) vs. $H_1: K=5$ (SBM) with increasing fifth community size $\delta\in[0.01, 0.20]$ and $\beta=0.8$. All networks generated from the alternative model.}
    \label{fig:sbmK5}
\end{figure}

\subsubsection{Number of communities in DCBM}

We use similar settings to test for the number of communities in a DCBM. For these results, we set the sparsity parameter to $\alpha=0.25$, sample $\psi_i\stackrel{\text{iid.}}{\sim}\mathsf{Uniform}(0.25, 0.75)$ and consider $\gamma\in\{0.4, 0.5,0.6\}$ for the proposed method. We choose values of $\gamma$ that are closer to 0.5 because fitting the CL model requires estimating $n$ parameters compared to the one parameter of the ER model, i.e., we require more data in $\bZ$. We also adopt the spherical clustering algorithm to detect the DCBM community labels as it is proven to consistently recover the true labels \citep{qin2013regularized, lei15}.

In the first setting, we test $H_0:K=1$ vs. $H_1:K=2$, noting that $K=1$ corresponds to the Chung-Lu model, i.e., degree heterogeneity but no community structure. We set $\boldsymbol{\pi}=(0.75, 0.25)^\top$ and vary $\beta\in\{0, 0.0475, \dots, 0.950\}$. When $\beta=0$, the networks are generated from the null model (CL) but for $\beta>0$, they come from the alternative model (DCBM). The results are in Figure \ref{fig:dcbmK2} (left). While ECV has the largest power, it also rejects 100\% of the time for $\beta=0$, i.e., when the null hypothesis is true. Thus, ECV is not reliable here because it does not maintain an appropriate type I error rate. The proposed method with $\gamma=0.4$ performs the best, with a power that begins to increase around $\beta=0.75$. NETCROP, however, does not reject the null hypothesis for any values of $\beta$.

In the next setting, we again test $H_0:K=1$ vs. $H_1:K=2$, now fixing $\beta=0.7$ and varying $\delta\in\{0.05, 0.10,\dots, 0.50\}$ where $\boldsymbol{\pi}=(\delta, 1-\delta)^\top$. The results are in Figure \ref{fig:dcbmK2} (right) and are similar to the previous setting. Even though all networks are generated from the alternative model and the ECV power is one, recall from the previous experiment that it does not maintain a proper type I error rate. The proposed method with $\gamma=0.4$ and $\gamma=0.5$ consistently yields the best rejection rate of the remaining methods, while NETCROP and the e-value method with $\gamma=0.6$ perform similarly. Finally, we test $H_0:K=4$ vs. $H_1:K=5$ with $\alpha=\beta=0.90$ and vary $\delta\in[0.01, 0.20]$ of length 20 where $\boldsymbol{\pi}=((1-\delta)/(K-1),\dots,(1-\delta)/(K-1),\delta)^\top$. The results are in Figure \ref{fig:dcbmK5}. ECV has the best power and the proposed methods with $\gamma\in\{0.4,0.5\}$ performing similarly.

\begin{figure}
    \centering
    \includegraphics[width=0.9\linewidth]{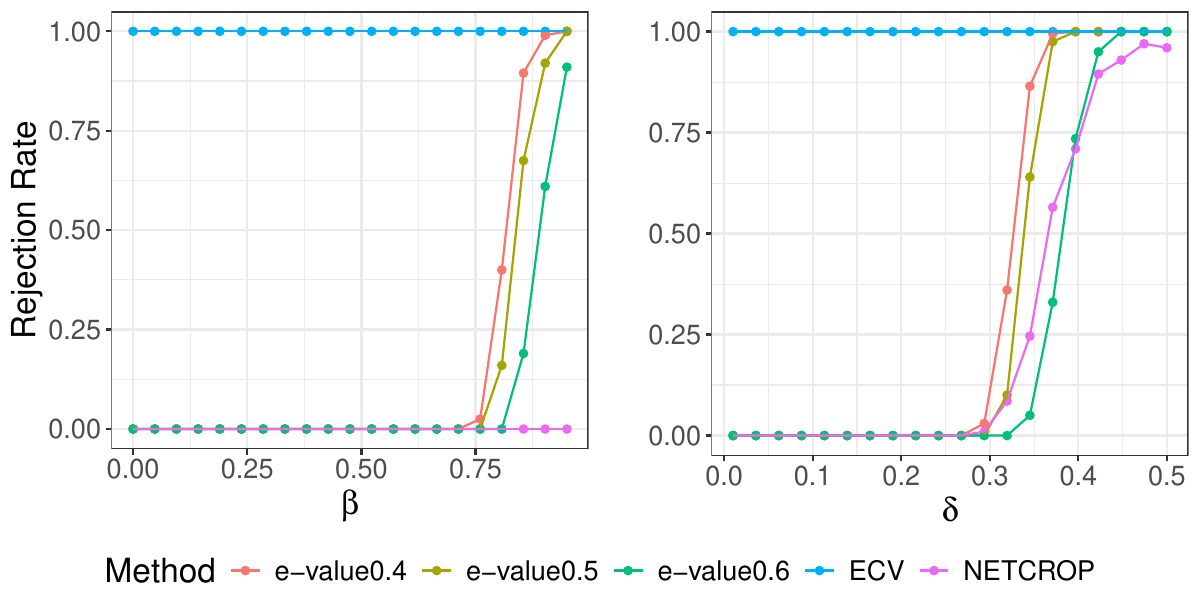}
    \caption{(left) $H_0: K=1$ (CL) vs. $H_1: K=2$ (DCBM) with increasing community strength $\beta\in[0,0.95]$, and $\delta=0.25$. $\beta=0$ corresponds to the null model while $\beta>0$ corresponds to the alternative model. (right) $H_0: K=1$ (CL) vs. $H_1: K=2$ (DCBM) with increasing community size $\delta\in[0.05, 0.5]$, and $\beta=0.7$. All networks generated from the alternative model. }
    \label{fig:dcbmK2}
\end{figure}

\begin{figure}
    \centering
    \includegraphics[width=0.55\linewidth]{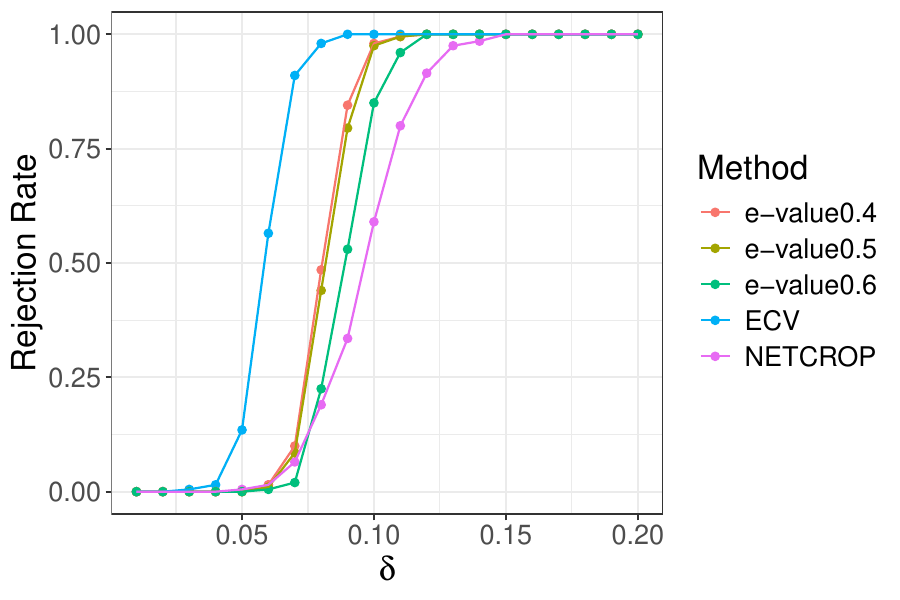}
    \caption{$H_0: K=4$ (DCBM) vs. $H_1: K=5$ (DCBM) with increasing fifth community size $\delta\in[0.01, 0.20]$, and $\beta=0.9$. All networks generated from the alternative model. }
    \label{fig:dcbmK5}
\end{figure}

\subsubsection{SBM versus DCBM}
In the previous two sub-sections, we tested for the number of communities but used the same block model in the null and alternative. Now, we want to test between two different block models, namely $H_0:$ SBM with $K$ communities vs. $H_1:$ DCBM with $K$ communities. To generate data, we sample DCBM networks where $\psi_i\stackrel{\text{iid.}}{\sim}\mathsf{Uniform}(0.5-\tfrac12\nu,0.5+\tfrac12\nu)$. Since the SBM is a special case of the DCBM where $\psi_i$ are equal for all $i$, we are effectively testing $H_0:\nu=0$ vs. $H_1:\nu>0$. In the first setting, we fix $\alpha=0.25$ and $\beta=0.50$ with $K=2$ and $\boldsymbol{\pi}=(0.75,0.25)^\top$. We then vary $\nu\in\{0.00, 0.04,\dots,0.80\}$ where larger $\nu$ means there is more degree heterogeneity among the nodes which should make it easier to reject the null hypothesis. When $\nu=0$, the DCBM reduces to the SBM so the null hypothesis is true, otherwise we expect to see a large rejection rate. We consider $\gamma\in\{0.4,0.5,0.6\}$ with the results in Figure \ref{fig:sbm_dcbm} (left). The e-value statistic with $\gamma=0.4$ yields the best results because although ECV has a similar power, it also has a large type I error rate $(\approx 0.15)$. In Figure \ref{fig:sbm_dcbm} (right), we plot the results for the same testing scenario but now with $\alpha=0.25$, $\beta=0.80$, $K=5$, $\boldsymbol{\pi}=(0.2,\dots,0.2)^\top$ and $\nu\in\{0.25, 0.30,\dots,0.75\}$. ECV consistently yields the largest power, followed by the proposed method with $\gamma=0.4$. 

Finally, in Section \ref{supp:sim} of the Supplemental Materials, we include a simulation setting testing the DCBM with $K$ communities against a PABM with $K$ communities. As neither ECV nor NETCROP are described for this setting, we compare the proposed method with varying values of $\gamma$. We find that the rejection rate increases as the data-generating model further departs from the null hypothesis.

\begin{figure}
    \centering
    \includegraphics[width=0.9\linewidth]{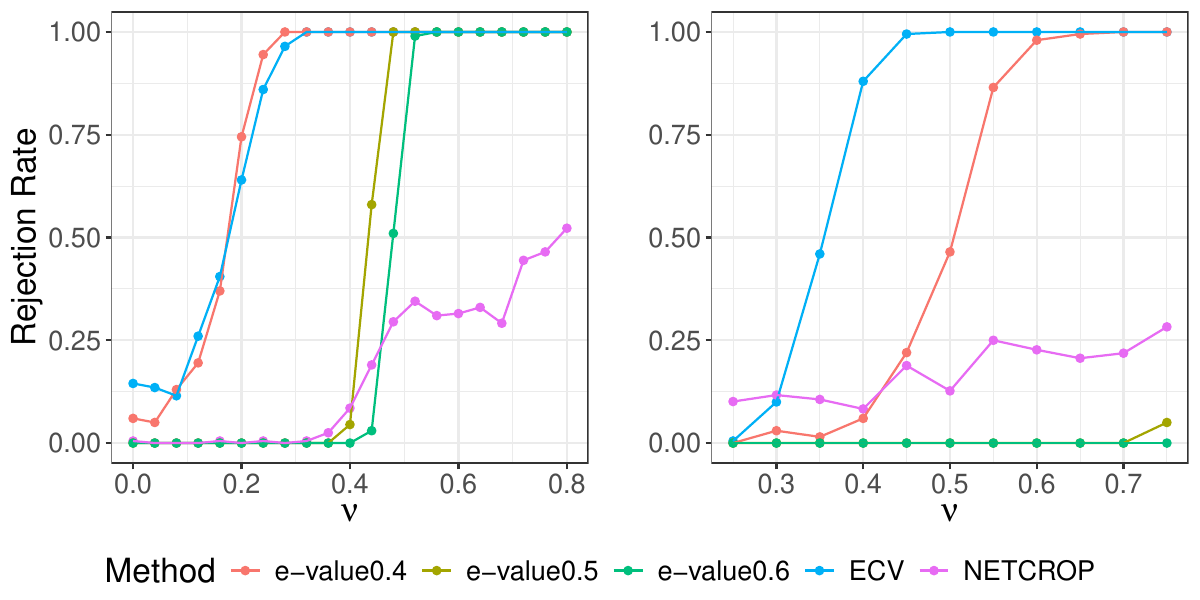}
    \caption{(left) $H_0: \nu=0$ (SBM) vs. $H_1: \nu>0$ (DCBM) with increasing degree variability $\nu\in[0,0.80]$ and $K=2$. $\nu=0$ corresponds to the null model while $\nu>0$ corresponds to the alternative model. (right) $H_0: \nu=0$ (SBM) vs. $H_1: \nu>0$ (DCBM) with increasing degree variability $\nu\in[0.25,0.75]$ and $K=5$. All networks generated from the alternative model.}
    \label{fig:sbm_dcbm}
\end{figure}

\subsection{Political Blogs network}\label{sec:real}
We demonstrate the performance of the proposed method on a real-world network. Over the past two decades, the political blogs network \citep{adamic2005political} has provided a popular case-study for new network methods, in part because of the known underlying community labels. Specifically, the nodes in this network correspond to blogs (websites) and there is an edge between nodes if the blogs reference one another. This network has $n=1224$ nodes and $m=16,715$ (undirected) edges. A natural partition of the nodes arises from labeling the blogs as liberal or conservative, making it a popular benchmark for various community detection algorithms. In particular, this network was studied in the original DCBM paper \citep{karrer2011stochastic} where it was shown the DCBM yielded a better fit than that of the SBM. 

We apply the proposed method to this network under multiple different hypotheses. Unless otherwise specified, we set $K=2$ since there are believed to be two communities corresponding to liberal and conservative blogs. Additionally, we set $\gamma=0.5$ and repeat the data splitting 1000 times with average e-values reported in Table \ref{tab:real}. First, we test the ER model against the CL model, i.e., testing if there is statistically significant degree heterogeneity in the network. As the ER model is very simplistic, it is unsurprising that this test yields such a large e-value. Similarly, a test of the ER model against a two-block SBM, effectively testing for the presence of a community structure, yields a large value of the test statistic, indicating a statistically significant community structure. In a parallel way, we can test the CL model against the DCBM, testing if degree heterogeneity alone is enough to model the network, or if degree heterogeneity plus community structure is required. The latter is affirmed with the large e-value. Next, we test SBM vs.~DCBM to see if block structure alone is sufficient, or if degree heterogeneity is also needed. This large e-value confirms the need for the degree parameters in the DCBM, further validating the conclusions in \cite{karrer2011stochastic}. We then test a DCBM with $K=2$ communities against a DCBM with $K=3$ to see if there are truly two communities (liberal and conservative) in the network. A small e-value gives strong evidence in favor of two communities. Lastly, we test the DCBM against the PABM to see if the extra flexibility afforded by the latter yields a better model fit. We again find a small e-value, implying that the popularity of a node is {\it not} different between different communities. In other words, a blog that is popular among liberal blogs is expected to have a similar popularity among conservative blogs. Thus, our final conclusion is that a DCBM with $K=2$ communities best fits this network, i.e. it exhibits statistically significant community structure and degree heterogeneity.

\begin{table}[]
    \centering
    \begin{tabular}{ll|c}
         $H_0$ & $H_1$ & e-value  \\\hline
         ER & CL & $>10^{100}$\\
         ER & SBM & $>10^{100}$\\
         CL & DCBM & $>10^{100}$\\
         SBM & DCBM & $>10^{100}$ \\
         DCBM $(K=2)$ & DCBM $(K=3)$ & 0\\
         DCBM & PABM & 0
    \end{tabular}
    \caption{Results from applying the proposed method to the political blogs network averaged over 1000 data splits.}
    \label{tab:real}
\end{table}

\section{Conclusion}\label{sec:conc}

In this work, we leveraged edge sampling and Universal Inference to carry out model selection on networks. We showed the proposed method's applicability across numerous testing scenarios including testing the number of communities and choice of random graph model. Our method comes with finite-sample type I error control and asymptotic growth-rate guarantees, while also demonstrating strong empirical performance. While we only applied our proposal to a few testing scenarios, we stress that the ideas can easily be extended to almost limitless settings, provided the data-generating mechanism has a likelihood. Extensions of this work to, e.g., testing on temporal networks \citep{yanchenko2025hypothesis} or settings without a likelihood \citep{dey2025generalized}, provide many interesting directions to pursue. Additionally, we acknowledge that a data-splitting approach inherently loses information, so future work can include directly constructing an e-value as it could lead to a more powerful test. Finally, our results showed that the edge sampling parameter, $\gamma$, plays a large role in the empirical performance of the proposed method. While we gave some general guidelines on how to set this parameter, a more thorough and perhaps theoretical treatment might be warranted, i.e., setting $\gamma$ such that the numerator and denominator of $E$ have the same variance. 

\section*{Acknowledgments}
We would like to thank Srijan Sengupta for helpful discussions. Claude was used to identify grammatical and typographical errors.

\bibliographystyle{apalike} 
\bibliography{refs}

@article{leiner2025data,
  title={Data fission: splitting a single data point},
  author={Leiner, James and Duan, Boyan and Wasserman, Larry and Ramdas, Aaditya},
  journal={Journal of the American Statistical Association},
  volume={120},
  number={549},
  pages={135--146},
  year={2025},
  publisher={Taylor \& Francis}
}

@article{yanchenko2025statistical,
  title={Statistical inference for core-periphery structures},
  author={Yanchenko, Eric and Sengupta, Srijan and Mukherjee, Diganta},
  journal={arXiv preprint arXiv:2508.04730},
  year={2025}
}

@article{wasserman2020universal,
  title={Universal inference},
  author={Wasserman, Larry and Ramdas, Aaditya and Balakrishnan, Sivaraman},
  journal={Proceedings of the National Academy of Sciences},
  volume={117},
  number={29},
  pages={16880--16890},
  year={2020},
  publisher={National Academy of Sciences}
}

@article{chakrabarty2025network,
  title={Network Cross-Validation and Model Selection via Subsampling},
  author={Chakrabarty, Sayan and Sengupta, Srijan and Chen, Yuguo},
  journal={arXiv preprint arXiv:2504.06903},
  year={2025}
}

@article{bhadra2025unified,
  title={A Unified Framework for Community Detection and Model Selection in Blockmodels},
  author={Bhadra, Subhankar and Tang, Minh and Sengupta, Srijan},
  journal={Journal of Computational and Graphical Statistics},
  number={just-accepted},
  pages={1--27},
  year={2025},
  publisher={Taylor \& Francis}
}

@article{bickel2016hypothesis,
  title={Hypothesis testing for automated community detection in networks},
  author={Bickel, Peter J and Sarkar, Purnamrita},
  journal={Journal of the Royal Statistical Society Series B: Statistical Methodology},
  volume={78},
  number={1},
  pages={253--273},
  year={2016},
  publisher={Oxford University Press}
}

@article{wang2017likelihood,
  title={Likelihood-based model selection for stochastic block models},
  author={Wang, Y. X. Rachel and Bickel, Peter J},
  journal={Annals of Statistics},
  volume={45},
  number={2},
  pages={500--528},
  year={2017}
}

@article{ancell2025post,
  title={Post-selection inference with a single realization of a network},
  author={Ancell, Ethan and Witten, Daniela and Kessler, Daniel},
  journal={arXiv preprint arXiv:2508.11843},
  year={2025}
}

@inproceedings{adamic2005political,
  title={The political blogosphere and the 2004 US election: divided they blog},
  author={Adamic, Lada A and Glance, Natalie},
  booktitle={Proceedings of the 3rd international workshop on Link discovery},
  pages={36--43},
  year={2005}
}

@article{yanchenko2024generalized,
  title={A generalized hypothesis test for community structure in networks},
  author={Yanchenko, Eric and Sengupta, Srijan},
  journal={Network Science},
  volume={12},
  number={2},
  pages={122--138},
  year={2024},
  publisher={Cambridge University Press}
}

@article{rohe11,
author={Rohe, Karl and Chatterjee, Sourav and Yu, Bin},
title={Spectral Clustering and The High-Dimensional Stochastic Blockmodel},
journal={The Annals of Statistics},
volume={39},
number={4},
year={2011},
pages={1878--1915}
}

@article{lei15,
author={Lei, Jing and Rinaldo, Alessandro},
title={Consistency of Spectral Clustering in Sparse Stochastic Block Models},
journal={Annals of Statistics},
year={2015},
volume={43},
number={1},
pages={215--237}
}

@Article{karrer2011stochastic,
  Title                    = {Stochastic blockmodels and community structure in networks},
  Author                   = {Karrer, B. and Newman, Mark E. J.},
  Journal                  = {Physical Review E},
  Year                     = {2011},
  Pages                    = {016107},
  Volume                   = {83},

  Publisher                = {APS}
}

@Article{holland1983stochastic,
  Title                    = {Stochastic blockmodels: first steps},
  Author                   = {Holland, P.W. and Laskey, K.B. and Leinhardt, S.},
  Journal                  = {Social Networks},
  Year                     = {1983},
  Pages                    = {109--137},
  Volume                   = {5},

  Publisher                = {Elsevier}
}

@article{Erdos1959,
  author  = {Erd{\"o}s, Paul and R{\'e}nyi, Alfr{\'e}d},
  title   = {On random graphs},
  journal = {Publicationes Mathematicae Debrecen},
  year    = {1959},
  volume  = {6},
  pages   = {290--297},
}

@article{chung2002average,
  title={The average distances in random graphs with given expected degrees},
  author={Chung, Fan and Lu, Linyuan},
  journal={Proceedings of the National Academy of Sciences},
  volume={99},
  number={25},
  pages={15879--15882},
  year={2002},
  publisher={National Acad Sciences}
}

@article {senguptapabm,
author = {Sengupta, Srijan and Chen, Yuguo},
title = {A block model for node popularity in networks with community structure},
journal = {Journal of the Royal Statistical Society: Series B (Statistical Methodology)},
 volume={80},
  number={2},
  pages={365--386},
  year={2018},
}

@article{bickel2009nonparametric,
  title={A nonparametric view of network models and Newman--Girvan and other modularities},
  author={Bickel, Peter J and Chen, Aiyou},
  journal={Proceedings of the National Academy of Sciences},
  volume={106},
  number={50},
  pages={21068--21073},
  year={2009},
  publisher={National Acad Sciences}
}

@article{li2020network,
  title={Network cross-validation by edge sampling},
  author={Li, Tianxi and Levina, Elizaveta and Zhu, Ji},
  journal={Biometrika},
  volume={107},
  number={2},
  pages={257--276},
  year={2020},
  publisher={Oxford University Press}
}

@article{yanchenko2025hypothesis,
  title={Hypothesis testing for community structure in temporal networks using e-values},
  author={Yanchenko, Eric and Williams, Jonathan P and Martin, Ryan},
  journal={Network Science},
  volume={e14},
  pages={1--17},
  year={2026}
}

@article{grunwald2024safe,
  title={Safe testing},
  author={Gr{\"u}nwald, Peter and de Heide, Rianne and Koolen, Wouter},
  journal={Journal of the Royal Statistical Society Series B: Statistical Methodology},
  volume={86},
  number={5},
  pages={1091--1128},
  year={2024},
  publisher={Oxford University Press UK}
}

@article{vovk2021values,
  title={E-values: Calibration, combination and applications},
  author={Vovk, Vladimir and Wang, Ruodu},
  journal={The Annals of Statistics},
  volume={49},
  number={3},
  pages={1736--1754},
  year={2021},
  publisher={Institute of Mathematical Statistics}
}

@article{wang2022false,
  title={False discovery rate control with e-values},
  author={Wang, Ruodu and Ramdas, Aaditya},
  journal={Journal of the Royal Statistical Society Series B: Statistical Methodology},
  volume={84},
  number={3},
  pages={822--852},
  year={2022},
  publisher={Oxford University Press}
}

@article{holme2012temporal,
  title={Temporal networks},
  author={Holme, Petter and Saram{\"a}ki, Jari},
  journal={Physics Reports},
  volume={519},
  number={3},
  pages={97--125},
  year={2012},
  publisher={Elsevier}
}

@inproceedings{leiner2024graph,
  title={Graph fission and cross-validation},
  author={Leiner, James and Ramdas, Aaditya},
  booktitle={International Conference on Artificial Intelligence and Statistics},
  pages={2638--2646},
  year={2024},
  organization={PMLR}
}

@article{qin2013regularized,
  title={Regularized spectral clustering under the degree-corrected stochastic blockmodel},
  author={Qin, Tai and Rohe, Karl},
  journal={Advances in neural information processing systems},
  volume={26},
  year={2013}
}

@article{park2025robust,
  title={Robust universal inference for misspecified models},
  author={Park, Beomjo and Balakrishnan, Sivaraman and Wasserman, Larry},
  journal={Biometrika},
  pages={asaf070},
  year={2025},
  publisher={Oxford University Press}
}

@article{dey2025generalized,
  title={Generalized universal inference on risk minimizers},
  author={Dey, Neil and Martin, Ryan and Williams, Jonathan P},
  journal={Journal of the Royal Statistical Society Series B: Statistical Methodology},
  pages={qkaf065},
  year={2025},
  publisher={Oxford University Press UK}
}

@article{hu2020corrected,
  title={Corrected Bayesian information criterion for stochastic block models},
  author={Hu, Jianwei and Qin, Hong and Yan, Ting and Zhao, Yunpeng},
  journal={Journal of the American Statistical Association},
  volume={115},
  number={532},
  pages={1771--1783},
  year={2020},
  publisher={Taylor \& Francis}
}

@article{de1998upper,
  title={An upper bound on the sum of squares of degrees in a graph},
  author={de Caen, Dominique},
  journal={Discrete Mathematics},
  volume={185},
  number={1-3},
  pages={245--248},
  year={1998},
  publisher={Elsevier}
}

\clearpage

\renewcommand{\thesection}{S\arabic{section}}
\setcounter{section}{0}

\begin{center}
    {\Large Supplemental Materials}
\end{center}

\begin{center}
    Table of Contents
\end{center}

\begin{enumerate}
    \item[S1] Joint distribution of $Z_{ij}$ and $Y_{ij}$
    \item[S2] Proof of Theorem \ref{thm:validity} (e-value)
    \item[S3] Proof of Theorem \ref{thm:er} (ER null)
    \item[S4] Proof of Proposition \ref{prop:ui} (expectation of log e-value)
    \item[S5] Proof of Theorem \ref{thm:cl} (CL null)
    \item[S6] Proof of Theorem \ref{thm:sbm} (SBM null)
    \item[S7] Additional simulation results: DCBM versus PABM
\end{enumerate}

\clearpage

\section{Joint distribution of $Z_{ij}$ and $Y_{ij}$}\label{supp:joint}
We now study the relationship between $Z_{ij}$ and $Y_{ij}$. Recall that $A_{ij}\stackrel{\text{ind.}}{\sim}\mathsf{Bernoulli}(\varrho_n P_{ij})$ where $P_{ij}\equiv P_{ij}(\btheta)$. If $A_{ij}=0$, which happens with probability $1-\varrho_n P_{ij}$, then $(Y_{ij},Z_{ij})=(0,0)$ as there is no edge to split. If $A_{ij}=1$, then $(Y_{ij},Z_{ij})=(0,1)$ with probability $\gamma$, and $(Y_{ij},Z_{ij})=(1,0)$ with probability $1-\gamma$. It is impossible for $(Y_{ij},Z_{ij})=(1,1)$. Thus, the joint distribution is defined by
\begin{align*}
    \mathbb P(Z_{ij}=0,Y_{ij}=0)&=1-\varrho_n P_{ij}\\
     \mathbb P(Z_{ij}=0,Y_{ij}=1)&=(1-\gamma)\varrho_n P_{ij}\\
    \mathbb P(Z_{ij}=1,Y_{ij}=0)&=\gamma \varrho_n P_{ij}\\
    \mathbb P(Z_{ij}=1,Y_{ij}=1)&=0,
\end{align*}
so marginally,
\begin{align*}
    Y_{ij}&\stackrel{\text{ind.}}{\sim}\mathsf{Bernoulli}((1-\gamma)\varrho_n P_{ij})\\
    Z_{ij}&\stackrel{\text{ind.}}{\sim}\mathsf{Bernoulli}(\gamma\varrho_n P_{ij}).
\end{align*}
It is easy to use these results to find the conditional distribution of $Z_{ij}\mid Y_{ij}$:
$$
    Z_{ij}\mid Y_{ij}\sim\mathsf{Bernoulli}\left(\frac{\gamma \varrho_n P_{ij}(1-Y_{ij})}{1-(1-\gamma)\varrho_n P_{ij}}\right).
$$
Now, in order to show independence, we need $\mathbb P(Z_{ij}=z,Y_{ij}=y)=\mathbb P(Z_{ij}=z)\mathbb P(Y_{ij}=y)$ for all $y,z$. We can see that
\begin{align*}
    \mathbb P(Z_{ij}=0)\mathbb P(Y_{ij}=0)&=\{1-\gamma\varrho_nP_{ij}\}\{1-(1-\gamma)\varrho_nP_{ij}\}\\
     \mathbb P(Z_{ij}=0)\mathbb P(Y_{ij}=1)&=\{1-\gamma\varrho_nP_{ij}\}(1-\gamma)\varrho_n P_{ij}\\
    \mathbb P(Z_{ij}=1)\mathbb P(Y_{ij}=0)&=\gamma \varrho_n P_{ij}\{1-(1-\gamma)\varrho_n P_{ij}\}\\
    \mathbb P(Z_{ij}=1)\mathbb P(Y_{ij}=1)&=\gamma(1-\gamma)\varrho_n^2 (P_{ij})^2  
\end{align*}
so $Z_{ij}$ and $Y_{ij}$ are clearly not independent. Since $\varrho_n\to0$ as $n\to\infty$, however, for large $n$, we can ignore terms that depend on $\varrho_n^2$. Thus, for large $n$,
\begin{align*}
    \mathbb P(Z_{ij}=0)\mathbb P(Y_{ij}=0)&\sim 1-\varrho_nP_{ij}\\
     \mathbb P(Z_{ij}=0)\mathbb P(Y_{ij}=1)&\sim (1-\gamma)\varrho_nP_{ij}\\
    \mathbb P(Z_{ij}=1)\mathbb P(Y_{ij}=0)&\sim \gamma \varrho_n P_{ij}\\
    \mathbb P(Z_{ij}=1)\mathbb P(Y_{ij}=1)&\sim 0,
\end{align*}
such that $Y_{ij}$ and $Z_{ij}$ are asymptotically independent.

\clearpage

\section{Proof of Theorem \ref{thm:validity}}\label{supp:validity}
\noindent
{\it Let $\bA\sim\bP(\btheta^*)$ and consider testing}
$$
    H_0:\btheta^*\in\Theta_0\text{ vs. }H_1:\btheta^*\in\Theta_1.
$$
{\it Under \ref{a:edge} and for any $n$, the statistic constructed in Section \ref{sec:steps} is an e-value i.e.,}
$$
    \sup_{\btheta\in\Theta_0}\mathbb E_{\btheta}E\leq 1.
$$

\noindent 
{\it Proof.}
Assume that $\bA\sim\bP(\btheta^*)$ where $\btheta^*\in\Theta_0$, i.e., the network was generated under the null hypothesis. Then
$$
    \mathbb E_{\btheta^*}\left\{\frac{\mathcal L(\frac{\gamma}{1-\gamma}\bP(\hat\btheta_1)|\bZ)}{\mathcal L(\bP(\hat\btheta_0)|\bZ)}\right\}
    \leq \mathbb E_{\btheta^*}\left\{\frac{\mathcal L(\frac{\gamma}{1-\gamma}\bP(\hat\btheta_1)|\bZ)}{\mathcal L(\gamma\bP(\btheta^*)|\bZ)}\right\}
$$
where $\hat\btheta_1=\arg\max_{\btheta\in\Theta_1}\mathcal L(\btheta|\bY)$ since $\hat\btheta_0=\arg\max_{\btheta\in\Theta_0}\mathcal L(\btheta|\bZ)$ such that $\mathcal L(\bP(\hat\btheta_0)|\bZ)\geq \mathcal L(\gamma\bP(\btheta^*)|\bZ)$. For notational simplicity, we write $P_{ij}(\btheta^*)\equiv P_{ij}^*$ and $P_{ij}(\hat\btheta_1)\equiv \hat P_{ij}$. To analyze the term on the right-hand side, we will use an iterated expectations argument as in \cite{wasserman2020universal}. Thus,
\begin{multline*}
    \mathbb E\left\{\frac{\mathcal L(\frac{\gamma}{1-\gamma}\hat\bP(\hat\btheta)|\bZ)}{\mathcal L(\gamma\bP(\btheta^*)|\bZ)}\Bigg|\bY\right\}\\
    =\sum_{z_{11}=0}^1\cdots\sum_{z_{nn}=0}^1\prod_{i<j}\frac{e^{-\frac{\gamma}{1-\gamma}\hat P_{ij}}(\frac{\gamma}{1-\gamma}\hat P_{ij})^{Z_{ij}}}{e^{-\gamma  P_{ij}^*}(\gamma P_{ij}^*)^{Z_{ij}}} 
    \times \left(\frac{\gamma  P_{ij}^*(1-Y_{ij})}{1-(1-\gamma) P_{ij}^*}\right)^{Z_{ij}}\left(1-\frac{\gamma  P_{ij}^*(1-Y_{ij})}{1-(1-\gamma) P_{ij}^*}\right)^{1-Z{ij}}
\end{multline*}
where the form of the statistic comes from Section 3.1 and 
$$
    Z_{ij}\mid Y_{ij}\sim\mathsf{Bernoulli}\left(\frac{\gamma  P_{ij}^*(1-Y_{ij})}{1-(1-\gamma) P_{ij}^*}\right)
$$
by Section \ref{supp:joint}. This expectation can be written in a more general form as
$$
    \sum_{x_1=0}^1\cdots\sum_{x_n=0}^1\prod_{i=1}^n a_{i}^{x_{i}}b_i^{1-x_i},
$$
where we ignore multiplicative factors that don't depend on $x_i$ ($Z_{ij}$). Since $n$ is finite, as are all terms, we can clearly re-write this expression as
$$
    \left(\sum_{x_1=0}^1 a_1^{x_1}b_1^{1-x_1}\right)\times \cdots\times \left(\sum_{x_n=0}^1 a_n^{x_n}b_n^{1-x_n}\right)
$$
and evaluating each summation yields
$$
    \left(a_1+b_1\right)\times \cdots\times \left(a_n+b_n\right)
    =\prod_{i=1}^n(a_i+b_i).
$$
Returning to our expression of interest, this result implies
\begin{multline*}
    \mathbb E\left\{\frac{\mathcal L(\frac{\gamma}{1-\gamma}\hat\bP(\hat\btheta)|\bZ)}{\mathcal L(\gamma\bP(\btheta^*)|\bZ)}\Bigg|\bY\right\}
    =e^{-\sum_{i<j}(\frac{\gamma}{1-\gamma}\hat P_{ij}-\gamma P_{ij}^*)}\prod_{i<j}\left\{1-\frac{\gamma P_{ij}^*(1-Y_{ij})}{1-(1-\gamma) P_{ij}^*}+\frac{\frac{\gamma}{1-\gamma}\hat P_{ij}(1-Y_{ij})}{1-(1-\gamma) P_{ij}^*}\right\}.
\end{multline*}
Since $1+x\leq e^x$ for all $x$, 
\begin{multline*}
    \mathbb E\left\{\frac{\mathcal L(\frac{\gamma}{1-\gamma}\hat\bP(\hat\btheta)|\bZ)}{\mathcal L(\gamma\bP(\btheta^*)|\bZ)}\Bigg|\bY\right\}
    \leq \exp\left[-\sum_{i<j}\left(\tfrac{\gamma}{1-\gamma}\hat P_{ij}-\gamma P_{ij}^*+\frac{\gamma P_{ij}^*(1-Y_{ij})}{1-(1-\gamma) P_{ij}^*}-\frac{\frac{\gamma}{1-\gamma}\hat P_{ij}(1-Y_{ij})}{1-(1-\gamma) P_{ij}^*}\right)\right],
\end{multline*}
and simplifying this expression yields
\begin{align*}
    \mathbb E\left\{\frac{\mathcal L(\frac{\gamma}{1-\gamma}\bP(\hat\btheta_1)|\bZ)}{\mathcal L(\gamma\bP(\btheta^*)|\bZ)}\right\}
    &=\mathbb E\left(\mathbb E\left\{\frac{\mathcal L(\frac{\gamma}{1-\gamma}\bP(\hat\btheta_1)|\bZ)}{\mathcal L(\gamma\bP(\btheta^*)|\bZ)}\Bigg|\bY\right\}\right)\\
    &\leq \mathbb E\exp\left(-\frac{\gamma}{1-\gamma}\sum_{i<j} \frac{\{Y_{ij}-(1-\gamma)P_{ij}^*\}\{\hat P_{ij}- (1-\gamma) P_{ij}^*\}}{1-(1-\gamma)P_{ij}^*}\right).
\end{align*}
Now, instead of looking at the expected value, we claim that the sum inside is deterministically $\geq0$; i.e.,
$$
\sum_{i<j} \frac{\{Y_{ij}-(1-\gamma)P_{ij}^*\}\{\hat P_{ij}- (1-\gamma) P_{ij}^*\}}{1-(1-\gamma)P_{ij}^*}\geq 0.
$$
Re-arranging some terms, this is equivalent to proving that
\begin{align*}
   &\sum_{i} \frac{\{Y_{i}-(1-\gamma)P_{i}^*\}\{\hat P_{i}- (1-\gamma) P_{i}^*\}}{1-(1-\gamma)P_{i}^*} \\
   &=\sum_{i} \frac{[(1-Y_{i})- (1-\{1-\gamma\}P_{i}^*)][(1-\hat P_{i})- (1-\{1-\gamma\}P_{i}^*)]}{1-(1-\gamma)P_{i}^*} \\
   &\equiv \sum_i \frac{(X_i-\beta_i)(\hat\beta_i-\beta_i)}{\beta_i}\\
   &\geq 0,    
\end{align*}
where we have changed the summation to a single index and defined $X_i\equiv 1-Y_i\sim\mathsf{Bernoulli}(\beta_i)$ with $\beta_i=1-(1-\gamma)P_i^*$. To show this last line, we will use the following useful facts. By \ref{a:edge},
$$
    \sum_i \hat \beta_i=\sum_i X_i.
$$
Additionally, since $\hat\btheta_1=\arg\max_{\btheta\in\Theta_1}\mathcal L(\btheta|\bY)$, then $\mathcal L( \bP(\hat\btheta_1)|\bY))\geq\mathcal L(\bP(\btheta)|\bY)$ for all $\btheta$, or equivalently, $\mathcal L( \bP(\hat\bbeta)|\bX))\geq\mathcal L(\bP(\bbeta)|\bX)$. This implies
$$
    \prod_{i}e^{-\hat \beta_i}\hat \beta_i^{X_i}\geq\prod_ie^{-\beta_i}\beta_i^{X_i}\iff 
    \sum_i(-\hat\beta_i+X_i\log\hat\beta_i)\geq\sum_i(-\beta_i+X_i\log\beta_i).
$$
Shuffling around the last inequality gives
$$
    \sum_i(\hat\beta_i-\beta_i)\leq\sum_i X_i\log\frac{\hat\beta_i}{\beta_i},
$$
and under \ref{a:edge}, this implies 
$$
    \sum_i(X_i-\beta_i)\leq\sum_i X_i\log\frac{\hat\beta_i}{\beta_i}.
$$
Dividing both sides by $\sum_j X_j$ and applying Jensen's inequality gives
$$
    \frac{\sum_i(X_i-\beta_i)}{\sum_j X_j}\leq \sum_i\frac{X_i}{\sum_j X_j}\log\frac{\hat\beta_i}{\beta_i}\leq \log\left(\sum_i\frac{X_i}{\sum_j X_j}\frac{\hat\beta_i}{\beta_i}\right).
$$
Noting that $\log x\leq x-1$ for all $x$, we get
$$
    1-\frac{\sum_i\beta_i}{\sum_i X_i}\leq \sum_i\frac{X_i}{\sum_jX_j}\frac{\hat\beta_i}{\beta_i}-1
$$
or, equivalently,
$$
   2-\frac{\sum_i\beta_i}{\sum_i X_i}\leq \sum_i\frac{X_i}{\sum_jX_j}\frac{\hat\beta_i}{\beta_i}.
$$
Multiplying both sides by $\sum_j X_j$ yields
$$
    2\sum_i X_i -\sum_i\beta_i\leq \sum_i\frac{X_i\hat\beta_i}{\beta_i}.
$$
Since $2\sum_i X_i=\sum_iX_i+\sum_i\hat\beta_i$, we get
$$
    0 \le \sum_i\left(\frac{X_i\hat\beta_i}{\beta_i}-X_i-\hat\beta_i+\beta_i\right)
    =\sum_i\frac{(X_i-\beta_i)(\hat\beta_i-\beta_i)}{\beta_i},
$$
as desired. Therefore,
$$
    \mathbb E\left\{\frac{\mathcal L(\frac{\gamma}{1-\gamma}\bP(\hat\btheta_1)|\bZ)}{\mathcal L(\gamma\bP(\btheta^*)|\bZ)}\right\}
    \leq \mathbb E\exp\left(-\frac{\gamma}{1-\gamma}\sum_{i<j} \frac{\{Y_{ij}-(1-\gamma)P_{ij}^*\}\{\hat P_{ij}- (1-\gamma) P_{ij}^*\}}{1-(1-\gamma)P_{ij}^*}\right)
    \leq 1
$$
as we hoped to show. $\square$

\clearpage 

\section{Proof of Theorem \ref{thm:er}}\label{supp:er}
{\it Let $\bA_n\sim \varrho_n\bP(\btheta^*)$ under \ref{a:sparse}. Assume that we are testing 
$$
    H_0:\btheta^*\in\Theta_{0}\text{ vs. }\btheta^*\in\Theta_1
$$
where $\Theta_{0}$ is the parameter space corresponding to the Erdos-Renyi model, i.e., $P_{ij}(\btheta^*)=p$ for all $(i,j)$, and $\Theta_1$ corresponds to the popularity-adjusted block model, i.e., $P_{ij}(\btheta^*)=\lambda_{ic_j^*}^*\lambda_{jc_i^*}^*$, under \ref{a:K}--\ref{a:pbar}. If $\btheta^*\in\Theta_1$ such that $\bA_n$ is generated from the alternative model, then there exists some $\epsilon>0$ such that, as $n\to\infty$, }
$$
    \mathbb P(\log E_n>\epsilon\varrho_nn^2)\to 1.
$$
\subsection*{Main proof}
Recall that 
$$
    \log E_n
    =\log \frac{\mathcal L(\tfrac{\gamma}{1-\gamma}\bP(\hat\btheta_1)|\bZ)}{\mathcal L(\bP(\hat\btheta_0)|\bZ)}
    =\log\frac{\mathcal L(\tfrac{\gamma}{1-\gamma}\bP(\hat\btheta_1)|\bZ)}{\prod_{ij}e^{-\hat p}(\hat p)^{Z_{ij}}}
$$
where $\hat p=n^{-2}\sum_{ij}Z_{ij}$ (the MLE under the null hypothesis), and, from now on, we take all products/summations over all $(i,j)$ instead of $\{i<j\}$ for notational simplicity. This effectively treats the network as directed and allowing for self-loops, but all results also hold for undirected networks, i.e., only considering $\{i<j\}$. We first define the following terms:
\begin{align*}
    B_n&=\log\frac{\prod_{ij} e^{-\gamma \varrho_n P_{ij}(\btheta^*)}}
    {\prod_{ij} e^{-\hat p}}\\
    C_n&=\log \frac{\prod_{ij}\{\gamma\varrho_nP_{ij}(\btheta^*)\}^{Z_{ij}}}{\prod_{ij}(\gamma\varrho_n\bar p)^{Z_{ij}}}\\
    D_n&=\log \frac{\prod_{ij}(\gamma\varrho_n\bar p)^{Z_{ij}}}{\prod_{ij}(\hat p)^{Z_{ij}}}\\
    L_n&=\log\frac{\mathcal L(\tfrac{\gamma}{1-\gamma}\bP(\hat\btheta_1)|\bZ)}{\mathcal L(\gamma\varrho_n\bP(\btheta^*)|\bZ)}
\end{align*}
where $\bar p=n^{-2}\sum_{ij}P_{ij}(\btheta^*)$. Then it is clear that
$$
    \log E_n
    = B_n + C_n + D_n + L_n.
$$
Notice that for any two random variables $X$ and $Y$,
\begin{align*}
    \mathbb P(X+Y\leq a)
    &=\mathbb P(X+Y\leq a,X\leq -b)+\mathbb P(X+Y\leq a,X>-b)\\
    &\leq P(X\leq -b) + \mathbb P(Y\leq a+b).
\end{align*}
By induction, it is easy to show that this implies
$$
    \mathbb P\left(\sum_{i=1}^d Y_i\leq a\right)
    \leq \sum_{i=1}^{d-1}\mathbb P(Y_i\leq -a_i) + \mathbb P\left(Y_d\leq a+\sum_{i=1}^{d-1}a_i\right).
$$
Thus,
\begin{align*}
    &\mathbb P(\log E_n\leq \epsilon\varrho_nn^2)\\
    &=\mathbb P(B_n + C_n + D_n + L_n\leq \epsilon\varrho_nn^2)\\
    &\leq \mathbb P(B_n\leq -\epsilon_1\varrho_nn^2) + \mathbb P(D_n\leq -\epsilon_2\varrho_nn^2)
    + \mathbb P(L_n\leq -\epsilon_3\varrho_nn^2) + \mathbb P\{C_n\leq (\epsilon+\epsilon_1+\epsilon_2+\epsilon_3)\varrho_nn^2\}
\end{align*}
where $\epsilon>0$ will be chosen later. The claim is proved if we can show that each of these probabilities goes to zero, which we do with the following lemmas.

\begin{lemma}\label{lemma:1}
    {\it Let $\bA\sim\varrho_n\bP(\btheta^*)$ where $\bP(\btheta^*)$ defines a popularity-adjusted block model (PABM), i.e., $P_{ij}(\btheta^*)=\lambda^*_{ic_j^*}\lambda_{jc_i^*}^*$, and $\hat\btheta=\arg\max_{\btheta\in\Theta}\mathcal L(\btheta|\bA)$. Define}
$$
    X_n=\log\frac{\prod_{ij} e^{-\varrho_nP_{ij}(\btheta^*)}}
    {\prod_{ij} e^{-\hat P_{ij}(\hat\btheta)}}.
$$
{\it Then}
$$
    (\varrho_nn^2)^{-1/2}X_n
    \stackrel{d}{\to}\mathsf{N}(0,k^2)
$$
{\it as $n\to\infty$ for some constant $k$, and, hence,} $X_n = o_p(\varrho_n n^2)$. 
\end{lemma}

\begin{lemma}\label{lemma:2}
   {\it Let $\bA\sim\varrho_n\bP(\btheta^*)$ where $\bP(\btheta^*)$ defines a popularity-adjusted block model (PABM), i.e., $P_{ij}(\btheta^*)=\lambda^*_{ic_j^*}\lambda_{jc_i^*}^*$. Furthermore, define $\hat p=n^{-2}\sum_{ij}A_{ij}$, $\bar p=n^{-2}\sum_{ij}P_{ij}(\btheta^*)$ and}
$$
    X_n=\log \frac{\prod_{ij}(\varrho_n\bar p)^{A_{ij}}}{\prod_{ij}(\hat p)^{A_{ij}}}.
$$
{\it Then for any $\epsilon>0$,}
$$
    \lim_{n\to\infty}\mathbb P(X_n\leq -\epsilon\varrho_nn^2)=0.
$$ 
\end{lemma}

\begin{lemma}\label{lemma:3}
  {\it Let $\bA_n\sim \varrho_n\bP(\btheta^*)$ where $\bP(\btheta^*)$ defines a popularity-adjusted block model (PABM), i.e., $P_{ij}(\btheta^*)=\lambda^*_{ic_j^*}\lambda_{jc_i^*}^*$. Define $\bY$ and $\bZ$ by the edge sampling procedure described in Section 3, $\hat\btheta_1=\arg\max_{\btheta\in\Theta}\mathcal L(\btheta|\bY)$ for any $\Theta$ in the PABM family where $\btheta^*\in\Theta$, and}
$$
    X_n=\log\frac{\mathcal L(\tfrac{\gamma}{1-\gamma}\bP(\hat\btheta_1)|\bZ)}{\mathcal L(\gamma\varrho_n\bP(\btheta^*)|\bZ)}. 
$$
{\it Then for any $\epsilon>0$,}
$$
    \lim_{n\to\infty}\mathbb P(X_n\leq -\epsilon\varrho_n n^2)=0.
$$  
\end{lemma}

\begin{lemma}\label{lemma:4}
    {\it Let}
$$
    f(\bu, \bv)
    =1-\frac{\sum_{i=1}^n (u_iv_i)^{1/2}}{\sum_{i=1}^n u_i}.
$$
{\it where $\sum_{i=1}^n u_i=\sum_{i=1}^n v_i$. Then $f(\bu,\bv)\geq0$ for all $\bu,\bv\geq0$ with equality if and only if $\bu$ and $\bv$ are linearly dependent.}
\end{lemma}

\begin{lemma}\label{lemma:5}
    {\it Let $\bA\sim\varrho_n\bP(\btheta_n^*)$ where $\bP(\btheta^*)$ defines a popularity-adjusted block model (PABM), i.e., $P_{ij}(\btheta^*)=\lambda^*_{ic_j^*}\lambda_{jc_i^*}^*$. Let}
    $$
        X_n
        =\log \frac{\prod_{ij}\{\varrho_nP_{ij}(\btheta^*)\}^{A_{ij}}}{\prod_{ij}(\varrho_n\bar p)^{A_{ij}}}.
    $$
    {\it where $ \bar p =n^{-2}\sum_{ij}P_{ij}(\btheta^*)$, i.e., the average edge probability. Then there exists some $\epsilon>0$ such that}
$$
    \lim_{n\to\infty}\mathbb P(X_n\leq \epsilon\varrho_n n^2)=0.
$$
\end{lemma}

\noindent
Thus, Lemmas \ref{lemma:1}, \ref{lemma:2} and \ref{lemma:3}, respectively, give us that, for any $\epsilon_1,\epsilon_2,\epsilon_3>0$,
$$
    \mathbb P(B_n<-\epsilon_1\varrho_nn^2)\to 0,
$$
$$
    \mathbb P(D_n\leq -\epsilon_2\varrho_nn^2)\to 0
$$
and 
$$
    \mathbb P(L_n\leq -\epsilon_3\varrho_nn^2)\to 0
$$
as $n\to\infty$. Finally, from Lemma \ref{lemma:5}, we know that we can set $\epsilon$ to be small enough such that
$$
    \epsilon+\epsilon_1+\epsilon_2+\epsilon_3
    <2\gamma \bar pC
$$
where
$$
    C=1-\frac{1}{\bar pn^2}\sum_{ij}\{\bar pP_{ij}(\btheta^*)\}^{1/2}>0.
$$
Thus, 
$$
    \mathbb P\{C_n\leq (\epsilon+\epsilon_1+\epsilon_2+\epsilon_3)\varrho_nn^2\}\to 0
$$
as $n\to\infty$. $\square$

\noindent
\subsection*{Proof of lemmas}
For all proofs of lemmas, we assume that \ref{a:sparse}--\ref{a:pbar} are met as in the statement of the theorem.\\

\noindent
{\bf Lemma \ref{lemma:1}.} {\it Let $\bA\sim\varrho_n\bP(\btheta^*)$ where $\bP(\btheta^*)$ defines a popularity-adjusted block model (PABM), i.e., $P_{ij}(\btheta^*)=\lambda^*_{ic_j^*}\lambda_{jc_i^*}^*$, and $\hat\btheta=\arg\max_{\btheta\in\Theta}\mathcal L(\btheta|\bA)$. Define}
        $$
                X_n=\log\frac{\prod_{ij} e^{-\varrho_nP_{ij}(\btheta^*)}}
                {\prod_{ij} e^{-\hat P_{ij}(\hat\btheta)}}.
        $$
        {\it Then}
        $$
            (\varrho_nn^2)^{-1/2}X_n
            \stackrel{d}{\to}\mathsf{N}(0,k^2)
        $$
        {\it as $n\to\infty$ for some constant $k$, and, hence,} $X_n = o_p(\varrho_n n^2)$. \\

        \noindent
        {\it Proof.} We have that
$$
      X_n
      =\sum_{ij} \{ P_{ij}(\hat\btheta)- \varrho_nP_{ij}(\btheta^*)\} \\
      =\sum_{ij}(A_{ij}-\mathbb E A_{ij})
$$
since \cite{senguptapabm} show that \ref{a:edge} is met for the PABM. The terms inside the sum are clearly independent with mean 0 and (finite) variance $ \varrho_n P_{ij}(\btheta^*)\{1- \varrho_n P_{ij}(\btheta^*)\}$ such that
$$
    s^2_n
    =\sum_{ij}  \varrho_n P_{ij}(\btheta^*)\{1- \varrho_n P_{ij}(\btheta^*)\}
    =\mathcal O(\varrho_n n^2).
$$
We want to invoke Lyapunov's Central Limit Theorem so we must show 
$$
    \lim_{n\to\infty} \frac1{s_n^{2+\delta}}\sum_{ij}\mathbb E|A_{ij}-\mathbb EA_{ij}|^{2+\delta}
    \to 0.
$$
Letting $\delta=2$ and noting that $\mathbb E|A_{ij}|^{4} = \varrho_nP_{ij}(\btheta^*)\sim \varrho_n$, we have that 
$$
     \lim_{n\to\infty} \frac1{s_n^{4}}\sum_{ij}\mathbb E|A_{ij}-\mathbb EA_{ij}|^{4}
     \lesssim \lim_{n\to\infty}\frac{1}{s_n^4}\sum_{ij}\varrho_n=\mathcal O(\tfrac{1}{\varrho_nn^2})
     \to 0
$$
by \ref{a:sparse}. Thus, by Lyapunov's Central Limit Theorem
$$
    \frac1{s_n}\sum_{ij}(A_{ij}-\mathbb EA_{ij})
    \stackrel{d}{\to}\mathsf {N}(0,1).
$$
Therefore, 
$$
    (\varrho_n n^2)^{-1/2}X_n
    =\left(\frac1{s_n}\sum_{ij}(A_{ij}-\mathbb EA_{ij})\right)\times \frac{s_n}{(\varrho_n n^2)^{1/2}}
    \stackrel{d}{\to}\mathsf{N}(0,k^2)
$$
since $s_n(\varrho_n n^2)^{-1/2}\to k$ for some constant $k$ as $n\to\infty$. $\square$\\

\noindent
{\bf Lemma \ref{lemma:2}.} {\it Let $\bA\sim\varrho_n\bP(\btheta^*)$ where $\bP(\btheta^*)$ defines a popularity-adjusted block model (PABM), i.e., $P_{ij}(\btheta^*)=\lambda^*_{ic_j^*}\lambda_{jc_i^*}^*$. Furthermore, define $\hat p=n^{-2}\sum_{ij}A_{ij}$, $\bar p=n^{-2}\sum_{ij}P_{ij}(\btheta^*)$ and}
$$
    X_n=\log \frac{\prod_{ij}(\varrho_n\bar p)^{A_{ij}}}{\prod_{ij}(\hat p)^{A_{ij}}}.
$$
{\it Then for any $\epsilon>0$,}
$$
    \lim_{n\to\infty}\mathbb P(X_n\leq -\epsilon\varrho_nn^2)=0.
$$
{\it Proof.} First, we write out $X_n$ as
$$
    X_n=\sum_{ij}A_{ij}\{\log(\varrho_n\bar p)-\log(\hat p)\}
    =n^2\hat p\log\frac{\varrho_n\bar p}{\hat p}.
$$
Dividing by $\varrho_nn^2$, we have
$$
    \frac{X_n}{\varrho_nn^2}=\frac{\hat p}{ \varrho_n}\log\frac{\bar p}{\hat p/(\varrho_n)}.
$$
Implicitly from Lemma \ref{lemma:1}, we know that $\hat p/\varrho_n\stackrel{p}{\to}\bar p$ since $\sum_{ij}(A_{ij}-\mathbb EA_{ij})=n^2(\hat p-\varrho_n\bar p)$ and by \ref{a:pbar} $\bar p$ does not depend on $n$. If we define $f(x)=x\log \frac{\bar p}{x}$, then by the continuous mapping theorem 
$$
    \frac{X_n}{\varrho_nn^2}
    = f\left(\frac{\hat p}{\varrho_n}\right)
    \stackrel{p}{\to}f(\bar p)
    =\bar p\log\frac{\bar p}{\bar p}
    =0.
$$
Thus, for any $\epsilon>0$, 
$$
    \lim_{n\to\infty}\mathbb P(X_n\leq -\epsilon\varrho_nn^2)=0.\ \ \ \square
$$

\noindent
{\bf Lemma \ref{lemma:3}.} {\it Let $\bA_n\sim \varrho_n\bP(\btheta^*)$ where $\bP(\btheta^*)$ defines a popularity-adjusted block model (PABM), i.e., $P_{ij}(\btheta^*)=\lambda^*_{ic_j^*}\lambda_{jc_i^*}^*$. Define $\bY$ and $\bZ$ by the edge sampling procedure described in Section 3, $\hat\btheta_1=\arg\max_{\btheta\in\Theta}\mathcal L(\btheta|\bY)$ for any $\Theta$ in the PABM family where $\btheta^*\in\Theta$, and}
$$
    X_n=\log\frac{\mathcal L(\tfrac{\gamma}{1-\gamma}\bP(\hat\btheta_1)|\bZ)}{\mathcal L(\gamma\varrho_n\bP(\btheta^*)|\bZ)}. 
$$
{\it Then for any $\epsilon>0$,}
$$
    \lim_{n\to\infty}\mathbb P(X_n\leq -\epsilon\varrho_n n^2)=0.
$$

\noindent
{\it Proof.} First, we note that we prove this result for a general PABM from which the results for the SBM (needed for Theorem 4) are easy to deduce as this model is a special case of the PABM.\\

\noindent
We start by defining
\begin{equation}\label{eq:lambda}
    \hat\lambda_{ik}(\hat\bc)
    =\frac{\sum_{j=1}^n Y_{ij}\mathds{1}(\hat c_j=k)}{\{\sum_{ab} Y_{ab}\mathds{1}(\hat c_a=\hat c_i, \hat c_b=k)\}^{1/2}}=\mathcal O_p(\varrho_n^{1/2})
\end{equation}
and
$$
    \lambda^*_{ik}(\hat\bc)
    =\frac{\sum_{j=1}^n P_{ij}(\btheta^*)\mathds{1}(\hat c_j=k)}{\{\sum_{ab} P_{ab}(\btheta^*)\mathds{1}(\hat c_a=\hat c_i, \hat c_b=k)\}^{1/2}}
$$
where $\hat\bc$ is a function of $\bY$ and \cite{senguptapabm} show that the MLEs of the PABM take the form in \eqref{eq:lambda}. We then write out $X_n$ as
\begin{align*}
    &\log\mathcal L(\tfrac{\gamma}{1-\gamma}\bP(\hat\btheta_1)|\bZ)-\log\mathcal L(\gamma\varrho_n\bP(\btheta^*)|\bZ)\\
    &=\sum_{ij}\{\tfrac{\gamma}{1-\gamma} P_{ij}(\hat\btheta_1)-\gamma \varrho_nP_{ij}(\btheta^*)\}
    +\sum_{ij}Z_{ij}\{\log \tfrac{\gamma}{1-\gamma} P_{ij}(\hat\btheta_1)-\log \gamma \varrho_nP_{ij}(\btheta^*)\}\\
    &= \tfrac{\gamma}{1-\gamma}\sum_{ij}(Y_{ij}-\mathbb EY_{ij})
    +\sum_{ij}Z_{ij}(\log \tfrac{\gamma}{1-\gamma}\hat\lambda_{i\hat c_j}\hat\lambda_{j\hat c_i}-\log \gamma\varrho_n\lambda^*_{ic_j^*}\lambda^*_{jc_i^*})\\
    &=\underbrace{\tfrac{\gamma}{1-\gamma}\sum_{ij}(Y_{ij}-\mathbb EY_{ij})}_{X_n^{(1)}}
    +2\underbrace{\sum_{ij}Z_{ij}\{\log \hat \lambda_{i\hat c_j}-\log \varrho_n^{1/2} \lambda^*_{i\hat c_j}(\hat\bc)\}}_{X_n^{(2)}}
    +2\underbrace{\sum_{ij}Z_{ij}\{\log \lambda^*_{i\hat c_j}(\hat\bc)-\log \lambda^*_{ic_j^*}\}}_{X_n^{(3)}} 
\end{align*}
where we used the fact that \ref{a:edge} is met, and for simplicity, we absorb the effect of $(1-\gamma)^{-1/2}$ into $\hat\lambda_{ik}$. Similar to the argument above, we have
\begin{align}\notag
    \mathbb P(X_n\leq -\epsilon\varrho_n n^2)
    &=\mathbb P(X_n^{(1)}+X_n^{(2)}+X_n^{(3)}\leq -\epsilon\varrho_n  n^2)\\ \label{eq:lemma3}
    &\leq \mathbb P(X_n^{(1)}\leq -\epsilon_1\varrho_n n^2) + \mathbb P(X_n^{(2)}\leq -\epsilon_2\varrho_n n^2) + \mathbb P(X_n^{(3)}\leq -(\epsilon-\epsilon_1-\epsilon_2)\varrho_n n^2).
\end{align}
We will show that each of these three probabilities goes to 0. First, by Lemma \ref{lemma:1}, for any $\epsilon_1>0$,
$$
    \mathbb P(X_n^{(1)}\leq -\epsilon_1\varrho_n n^2)\to 0.
$$
We next note that there could be an identifiability problem between $\hat\bc$ and $\bc^*$, i.e., community `1' in $\hat\bc$ could be different from community `1' in $\bc^*$. For notational simplicity, we assume that $\hat\bc$ and $\bc^*$ are ``aligned'' such that community labels correspond to the same blocks. Now, we separate the probability of interest into two cases, where $X_n^{(2)}$ is finite versus infinite:
$$
    \mathbb P(X_n^{(2)}\leq -\epsilon_2 n^2)
    =\mathbb P(X_n^{(2)}\leq -\epsilon_2 n^2| X_n^{(2)}>-\infty) \mathbb P(X_n^{(2)}>-\infty) + \mathbb P(X_n^{(2)}=-\infty).
$$
We break the probability up in this way because there is a small, but non-zero probability that $X_n^{(2)}=-\infty$. To see this, note that
$$
    X_n^{(2)}
    =\sum_{ij}Z_{ij}\{\log \hat \lambda_{i\hat c_j}-\log \varrho_n^{1/2} \lambda^*_{i\hat c_j}(\hat\bc)\}
    =\sum_{i=1}^n\sum_{k=1}^K M_{ik}\{\log \hat \lambda_{ik}-\log \varrho_n^{1/2} \lambda^*_{ik}(\hat\bc)\}
$$
where
$$
    M_{ik}
    =\sum_{j=1}^nZ_{ij}\mathds{1}(\hat c_j=k),
$$
i.e., the number of edges in $\bZ$ between node $i$ and nodes assigned to group $k$. It is important to remember that $M_{ik}$ is a function of $\bZ$, while $\hat \lambda_{ik}(\hat\bc)$ and $\hat\bc$ are estimated from $\bY$. We want to see when this term will equal negative infinity. First, $\lambda^*_{ik}(\hat\bc)>0$, for all $(i,k)$ by \ref{a:pmin}. If $M_{ik}=\hat\lambda_{ik}(\hat\bc)=0$, we can simply set this term in the sum to zero. The issue arises when $\hat \lambda_{ik}(\hat\bc)=0$ but $M_{ik}>0$. Then this term will be undefined, e.g., equal negative infinity. Thus, we need to compute
\begin{align*}
    \mathbb P\{\hat \lambda_{ik}(\hat\bc)=0, M_{ik}>0\}
    &=\mathbb P\{M_{ik}>0|\hat \lambda_{ik}(\hat\bc)=0)\mathbb P(\hat \lambda_{ik}(\hat\bc)=0\}\\
    &=\mathbb P\left\{\sum_{j=1}^nZ_{ij}\mathds{1}(\hat c_j=k)>0\Bigg|\sum_{j=1}^nY_{ij}\mathds{1}(\hat c_j=k)=0\right\}\\
    &\times \mathbb P\left\{\sum_{j=1}^nY_{ij}\mathds{1}(\hat c_j=k)=0\right\}
\end{align*}
Using the joint (conditional) distributions of $\bY$ and $\bZ$ found in Section \ref{supp:joint}, we can show
\begin{align*}
    &\mathbb P\left\{\sum_{j=1}^nZ_{ij}\mathds{1}(\hat c_j=k)>0\Bigg|\sum_{j=1}^nY_{ij}\mathds{1}(\hat c_j=k)=0\right\}\\  
    &=1-\mathbb P\left\{\sum_{j=1}^nZ_{ij}\mathds{1}(\hat c_j=k)=0\Bigg|\sum_{j=1}^nY_{ij}\mathds{1}(\hat c_j=k)=0\right\}\\
    &=1-\prod_{j=1}^n\left\{1-\frac{\gamma \lambda^*_{i\hat c_j}\lambda^*_{j\hat c_i}}{1-(1-\gamma)\lambda^*_{i\hat c_j}\lambda^*_{j\hat c_i}}\right\}^{\mathds{1}(\hat c_j=k)}\\
    &\leq 1.
\end{align*}
Additionally,
$$
    \mathbb P\left\{\sum_{j=1}^nY_{ij}\mathds{1}(\hat c_j=k)=0\right\}
    =\prod_{j=1}^n\{1-(1-\gamma) \lambda^*_{i\hat c_j}\lambda^*_{j\hat c_i}\}^{\mathds{1}(\hat c_j=k)}\}
    \leq (1- \delta)^{\hat n_k}
$$
where $\hat n_{k}=\sum_{j=1}^n \mathds{1}(\hat c_j=k)$ is the number of nodes assigned to community $k\in\{1,\dots,K\}$, $\delta=\min_{i,j,k,\ell}\{(1-\gamma)\lambda_{ik}^*\lambda_{j\ell}^*\}$ and $\delta>0$ by \ref{a:pmin}.
Thus,
$$
    \mathbb P\{\hat \lambda_{ik}(\hat\bc)=0, M_{ik}>0\}
    \leq (1- \delta)^{\hat n_{k}}
$$
such that
\begin{align*}
    \mathbb P(X_n^{(2)}=-\infty)
    &=\mathbb P\left(\bigcup_{i=1}^n\bigcup_{k=1}^K\{\hat \lambda_{ik}(\hat\bc)=0, M_{ik}>0\}\right)\\
    &\leq \sum_{i=1}^n\sum_{k=1}^K (1- \delta)^{\hat n_{k}}\\
    &\sim n(1-\delta)^{n}\\
    &\to 0
\end{align*}
for any $\delta>0$ as $n\to\infty$ where $\hat n_{k}=\mathcal O(n)$ by \ref{a:K}. Finally, while this probability is small and asymptotically goes to 0, it is non-zero for any $n$ which is why $\mathbb E\log E_n=-\infty$.\\

\noindent
Returning to the probability that we are currently interested in, we now have
$$
    \mathbb P(X_n^{(2)}\leq -\epsilon_2n^2)
    \leq \mathbb P(|X_n^{(2)}|\geq \epsilon_2 n^2\mid |X_n^{(2)}|<\infty)+ o(1).
$$
We showed above that
$$
    X_n^{(2)}
    =\sum_{i=1}^n\sum_{k=1}^K M_{ik}\{\log \hat \lambda_{ik}-\log \varrho_n^{1/2} \lambda^*_{ik}(\hat\bc)\}
$$
so by the Cauchy-Schwarz inequality,
$$
    (X_n^{(2)})^2
    \leq \left(\sum_{i=1}^n\sum_{k=1}^KM_{ik}^2\right)\left(\sum_{i=1}^n\sum_{k=1}^K\{\log \hat \lambda_{ik}-\log \varrho_n^{1/2} \lambda^*_{ik}(\hat\bc)\}^2\right).
$$
Recall that
$$
    M_{ik}=\sum_{i=1}^n\sum_{k=1}^K Z_{ij}\mathds{1}(\hat c_j=k).
$$
Since $\sum_{i=1}^n \mathds{1}(\hat c_j=k)=\mathcal O(n)$ for all $k\in\{1,\dots,K\}$ by \ref{a:K}, this implies
$$
    M_{ik}
    =\sum_{i=1}^n\sum_{k=1}^K Z_{ij}\mathds{1}(\hat c_j=k)
    \sim \sum_{i=1}^n Z_{ij}
    =D_i
$$
where $D_i$ is the degree of node $i$. In other words, $M_{ik}=\mathcal O_p(D_i)$ such that $\sum_{i=1}^n\sum_{k=1}^KM_{ik}^2\sim\sum_{i=1}^n D_i^2$. \cite{de1998upper} shows that
$$
    \sum_{i=1}^nD_i^2
    \leq m\left(\frac{2m}{n-1}+n+2\right)
$$
where $m$ is the total number of edges in the network, i.e., $m=\sum_{ij}Z_{ij}$. Clearly, $m=\mathcal O_p(\varrho_nn^2)$ such that
$$
    \sum_{i=1}^nD_i^2
    \leq m\left(\frac{2m}{n-1}+n+2\right)
    =\mathcal O_p(\varrho_n^2n^3)
$$
so
$$
    \sum_{i=1}^n\sum_{k=1}^KM_{ik}^2
    =\mathcal O_p(\varrho_n^2n^3).
$$
Next, we know that $\varrho_n^{-1/2}\hat\lambda_{ik}(\hat\bc)=\mathcal O_p(1)$ and is strictly positive since we conditioned on the fact that $|X_n^{(2)}|<\infty$ which implies $\hat\lambda_{ik}>0$. We also know that $\lambda^*_{ik}(\hat\bc)>0$ by \ref{a:pmin}. Therefore, $\log(x)$ is Lipschitz continuous on $[\delta,\infty)$ where $\delta=\min_{ik}\{\varrho_n^{-1/2}\hat\lambda_{ik}(\hat\bc), \lambda^*_{ik}(\hat\bc)\}>0$ such that
$$
    \sum_{i=1}^n\sum_{k=1}^K\{\log \hat \lambda_{ik}-\log \varrho_n^{1/2} \lambda^*_{ik}(\hat\bc)\}^2
    \lesssim \sum_{i=1}^n\sum_{k=1}^K\{ \hat \lambda_{ik}-  \varrho_n^{1/2}\lambda^*_{ik}(\hat\bc)\}^2
    =o_p(n)
$$
where the rate comes from the proof of Theorem \ref{thm:er} in \cite{senguptapabm} (with \ref{a:sparse}, \ref{a:K}, \ref{a:ident} and \ref{a:set}). Combining these two results implies that $X_n^{(2)}=o_p(\varrho_nn^2)$  such that for any $\epsilon_2>0$,
$$
    \mathbb P(X_n^{(2)}\leq-\epsilon_2\varrho_n n^2)
    \leq \mathbb P(|X_n^{(2)}|\geq \epsilon_2\varrho_n n^2\mid |X_n^{(2)}|<\infty) + \mathbb P(X_n^{(2)}=\infty)
    \to 0.
$$
For $X_n^{(3)}$, recall that
$$
    X_n^{(3)}
    =\sum_{ij}Z_{ij}\{\log\lambda^*_{i\hat c_j}(\hat\bc)-\log \lambda^*_{ic_j^*}\}.
$$
Notice that $\log\lambda^*_{i\hat c_j}(\hat\bc)-\log \lambda^*_{ic_j^*}=\mathcal O(1)$ and is non-zero if and only if $\hat c_j\neq c_j^*$. Thus,
$$
    X_n^{(3)}
    \sim \sum_{ij}Z_{ij}\mathds{1}(\hat c_j\neq c_j^*)
    =\sum_{i=1}^n D_i\mathds{1}(\hat c_i\neq c_i^*)
$$
where $D_i$ is again the degree of node $i$. Invoking Cauchy-Schwarz once more, 
$$
    (X_n^{(3)})^2
    \leq \left(\sum_{i=1}^nD_i^2\right)\left(\sum_{i=1}^n\mathds{1}(\hat c_i\neq c_i^*)\right).
$$
We already showed that $\sum_{i=1}^nD_i^2=\mathcal O_p(\varrho_n^2n^3)$, and, by Theorem 1 in \cite{senguptapabm} (with \ref{a:sparse}, \ref{a:K} and \ref{a:ident}), $\sum_{i=1}^n\mathds{1}(\hat c_i\neq c_i^*)=o_p(n)$, which implies $X_n^{(3)}=o_p(\varrho_nn^2)$.
Thus, for any $\epsilon_3>0$,
$$
    \mathbb P(X_n^{(3)}\leq -\epsilon_3\varrho_nn^2)
    \to 0
$$
as $n\to\infty$. $\square$\\

\noindent
{\bf Lemma \ref{lemma:4}.} {\it Let}
    $$
        f(\bu, \bv)
        =1-\frac{\sum_{i=1}^n (u_iv_i)^{1/2}}{\sum_{i=1}^n u_i}.
    $$
    {\it where $\sum_{i=1}^n u_i=\sum_{i=1}^n v_i$. Then $f(\bu,\bv)\geq0$ for all $\bu,\bv\geq0$ with equality if and only if $\bu$ and $\bv$ are linearly dependent.}\\

    \noindent
    {\it Proof.} By Cauchy-Schwarz,
    $$
        \sum_{i=1}^n (u_iv_i)^{1/2}
        \leq\left\{\left(\sum_{i=1}^n u_i\right)\left(\sum_{i=1}^n v_i\right)\right\}^{1/2}
        =\sum_{i=1}^n u_i
    $$
    since $\sum_{i=1}^n u_i=\sum_{i=1}^n v_i$ and with equality if and only if $\bu$ and $\bv$ are linearly dependent. The result follows immediately by dividing both sides of the inequality by $\sum_{i=1}^n u_i$. $\square$\\

\noindent
{\bf Lemma \ref{lemma:5}.} {\it Let $\bA\sim\varrho_n\bP(\btheta_n^*)$ where $\bP(\btheta^*)$ defines a popularity-adjusted block model (PABM), i.e., $P_{ij}(\btheta^*)=\lambda^*_{ic_j^*}\lambda_{jc_i^*}^*$. Let}
    $$
        X_n
        =\log \frac{\prod_{ij}\{\varrho_nP_{ij}(\btheta^*)\}^{A_{ij}}}{\prod_{ij}(\varrho_n\bar p)^{A_{ij}}}.
    $$
    {\it where $ \bar p =n^{-2}\sum_{ij}P_{ij}(\btheta^*)$, i.e., the average edge probability. Then there exists some $\epsilon>0$ such that}
$$
    \lim_{n\to\infty}\mathbb P(X_n\leq \epsilon\varrho_n n^2)=0.
$$

\noindent
{\it Proof.} We start by writing out $X_n$ as
$$
    X_n=\sum_{ij}A_{ij}\log\left(\frac{\varrho_nP_{ij}(\btheta^*)}{\varrho_n\bar p}\right)
$$
Recall that $A_{ij}\stackrel{\text{ind.}}{\sim}\mathsf{Bernoulli}(\varrho_nP_{ij}(\btheta^*))$.  If $Y\sim\mathsf{Bernoulli}(\theta)$, then the moment generating function, $M_Y(t)$, is
$$
    M_Y(t)
    =\mathbb E e^{Yt}
    =1-\theta+\theta e^t.
$$
Moreover, the moment generating function of $\alpha Y$ where $\alpha$ is a constant is
$$
    M_{\alpha Y}(t)
    =\mathbb{E} e^{(\alpha Y)t}
    =M_{Y}(\alpha t)
    =1-\theta+\theta e^{\alpha t}.
$$
Additionally, if $Y_1,\dots,Y_m$ are independent, then the moment generating function of their sum is simply the product, i.e.,
$$
    M_{\sum_{i=1}^M Y_i}(t)=\mathbb Ee^{t\sum_{i=1}^mY_i}=\prod_{i=1}^m\mathbb Ee^{tY_i}=\prod_{i=1}^m M_{Y_i}(t).
$$
Putting this together, the moment generating function of $X_n$ is
\begin{align*}
    M_{X_n}(t)
    &=\prod_{ij}\left\{1- \varrho_nP_{ij}(\btheta^*) +  \varrho_nP_{ij}(\btheta^*)e^{t\log\left(\frac{\varrho_nP_{ij}(\btheta^*)}{\varrho_n\bar p}\right)}\right\}\\
    &=\prod_{ij}\left[1+  \varrho_nP_{ij}(\btheta^*)\left\{\left(\frac{\varrho_nP_{ij}(\btheta^*)}{\varrho_n\bar p}\right)^t-1\right\}\right]\\
    &\leq \exp\left(\sum_{ij}\varrho_nP_{ij}(\btheta^*)\left\{\left(\frac{\varrho_nP_{ij}(\btheta^*)}{\varrho_n\bar p}\right)^t-1\right\}\right)
\end{align*}
where the last line is because $1+x\leq \exp(x)$ for all $x$. Then using a Chernoff bound argument, we have
\begin{align*}
    \mathbb P(X_n\leq \epsilon\varrho_n n^2)
    &\leq \inf_{t\leq 0}\exp(-t\epsilon\varrho_n n^2)M_{X_n}(t)\\
    &\leq \inf_{t\leq 0}\exp(-t\epsilon\varrho_n n^2)
    \exp\left(-\sum_{ij} \varrho_nP_{ij}(\btheta^*)\left\{1-\left(\frac{\varrho_nP_{ij}(\btheta^*)}{\varrho_n\bar p}\right)^t\right\}\right).
\end{align*}
Setting $t=-1/2$,
\begin{align*}
    \mathbb P(X_n\leq \epsilon\varrho_n n^2)
    &\leq \exp\left(\tfrac12\epsilon\varrho_n n^2 -  \sum_{ij}\left[ \varrho_nP_{ij}(\btheta^*)\left\{1-\left(\frac{\varrho_n\bar p}{\varrho_nP_{ij}(\btheta^*)}\right)^{1/2}\right\}\right]
    \right)\\
    &= \exp\left[\tfrac12\epsilon\varrho_n n^2 -  \sum_{ij}\left\{\varrho_nP_{ij}(\btheta^*)-\{\varrho_n^2\bar pP_{ij}(\btheta^*)\}^{1/2}\right\}
    \right]\\
    &= \exp\left(\bar p\varrho_nn^2\left[\frac{\epsilon}{2\bar p} -  \frac{1}{\bar pn^2}\sum_{ij}P_{ij}(\btheta^*)-\{\bar pP_{ij}(\btheta^*)\}^{1/2}\right]
    \right).
\end{align*}
Clearly,
$$
    \sum_{ij}P_{ij}(\btheta^*)=\sum_{ij}\bar p=\bar pn^2
$$
so by Lemma \ref{lemma:4},
$$
    \frac{1}{\bar pn^2}\sum_{ij}P_{ij}(\btheta^*)-\{\bar pP_{ij}(\btheta^*)\}^{1/2}
    =1 - \frac1{\bar p n^2}\sum_{ij} \{\bar pP_{ij}(\btheta^*)\}^{1/2}
    =C
$$
for some $C>0$ where the inequality is strict as the $P_{ij}(\btheta^*)$'s are linearly independent of (not all equal to) $\bar p$ by \ref{a:ident}. We can set $\epsilon>0$ small enough such that
$$
    \frac{\epsilon}{2\bar p}-C <0\implies \epsilon< 2 \bar pC
$$
where we recall that $\bar p$ does not depend on $n$ by \ref{a:ident}. Then we have
$$
    \mathbb P(X_n\leq\epsilon\varrho_n n^2)
    \lesssim \exp(- \varrho_nn^2)
    \to 0
$$
as $n\to\infty$ by \ref{a:sparse}. $\square$

\clearpage

\section{Proof of Proposition \ref{prop:ui}}\label{supp:ui}
{\it Let $X_1,\dots,X_{2n}\stackrel{\text{iid.}}{\sim}\mathsf{Bernoulli}(\theta)$ and $E_n$ be the split-likelihood Universal Inference statistic from \cite{wasserman2020universal}. Assume that we are testing}
$$
    H_0:\theta=\theta_0\text{ vs. }\theta\neq \theta_0.
$$
{\it where $0<\theta_0<1$. Then $\mathbb E(\log E_n)=-\infty$.}\\

\noindent
{\it Proof.} From \cite{wasserman2020universal}, we split $X_1,\dots,X_{2n}$ into two parts $D_0$ and $D_1$. For simplicity let $D_0=\{X_1,\dots,X_n\}$ and $D_1=\{X_{n+1},\dots,X_{2n}\}$. Then under the alternative hypothesis, the maximum likelihood estimate is
$$
    \hat\theta_1=\arg\max_{\theta}\mathcal L(\theta|D_1)=\frac1n\sum_{i=n+1}^{2n}X_i.
$$
Thus
$$
    \log E_n
    =\log\frac{\mathcal L(\hat\theta_1|D_0)}{\mathcal L(\theta_0|D_0)}
    =X\log\frac{\hat\theta_1}{\theta_0} + (n-X)\log\frac{1-\hat\theta_1}{1-\theta_0}
$$
where $X=\sum_{i=1}^n X_i$. Clearly, $\log E_n=-\infty$ if $\hat\theta_1=0$ and $X>0$. The same problem arises if $\hat\theta_1=1$ and $n-X>0$. Thus,
\begin{align*}
    \mathbb P(\log E_n=-\infty)
    &=\mathbb P(\hat\theta_1=0,X>0) + \mathbb P(\hat\theta_1=1,X<n)\\
    &=\mathbb P(\hat\theta_1=0)\{1-\mathbb P(X=0)\} + \mathbb P(\hat\theta_1=1)\{1-\mathbb P(X=n)\}\\
    &=(1-\theta)^n\{1-(1-\theta)^n\} + \theta^n(1-\theta^n)\\
    &>0.
\end{align*}
Since $\log E_n$ has a non-zero probability of being $-\infty$, its expectation necessarily is $-\infty$. $\square$

\clearpage

\section{Proof of Theorem \ref{thm:cl}}\label{supp:cl}
{\it Let $\bA_n\sim \varrho_n\bP(\btheta^*)$ under \ref{a:sparse}. Assume that we are testing 
$$
    H_0:\btheta^*\in\Theta_{0}\text{ vs. }\btheta^*\in\Theta_1
$$
where $\Theta_{0}$ is the parameter space corresponding to the Chung-Lu model, i.e., $P_{ij}(\btheta^*)=\psi_i\psi_j$, and $\Theta_1$ corresponds to the popularity-adjusted Block Model, i.e., $P_{ij}(\btheta^*)=\lambda_{ic_j^*}^*\lambda_{jc_i^*}^*$, under \ref{a:K}--\ref{a:cl}. If $\btheta^*\in\Theta_1$ such that $\bA_n$ is generated from the alternative model, then there exists some $\epsilon>0$ such that, as $n\to\infty$, }
$$
    \mathbb P(\log E_n>\epsilon\varrho_nn^2)\to 1.
$$

\subsection*{Main proof}
We follow a similar strategy as that of Theorem \ref{thm:er}, separating $\log E_n$ into the following terms:
\begin{align*}
    B_n&=\log\frac{\prod_{ij} e^{-\gamma \varrho_n P_{ij}(\btheta^*)}}
    {\prod_{ij} e^{-\hat \psi_i\hat\psi_j}}\\
    C_n&=\log \frac{\prod_{ij}\gamma\varrho_nP_{ij}(\btheta^*)^{Z_{ij}}}{\prod_{ij}(\gamma\varrho_n\tilde\psi^*_i\tilde\psi^*_j)^{Z_{ij}}}\\
    D_n&=\log \frac{\prod_{ij}(\gamma\varrho_n\tilde\psi^*_i\tilde\psi^*_j)^{Z_{ij}}}{\prod_{ij}(\hat \psi_i\hat\psi_j)^{Z_{ij}}}\\
    L_n&=\log\frac{\mathcal L(\tfrac{\gamma}{1-\gamma}\bP(\hat\btheta_1)|\bZ)}{\mathcal L(\gamma\varrho_n\bP(\btheta^*)|\bZ)}
\end{align*}
where 
$$
    \hat\psi_i
    =\frac{\sum_{j=1}^n Z_{ij}}{\sqrt{\sum_{ab}Z_{ab}}}
$$
and
$$
    \tilde\psi^*_i=\frac{\sum_{j=1}^nP_{ij}(\btheta^*)}{\sqrt{\sum_{ab}P_{ab}(\btheta^*)}}
$$
for $i\in\{1,\dots,n\}$. Then the result follows the same logic as the proof of Theorem \ref{thm:er} by invoking Lemmas \ref{lemma:1}, \ref{lemma:3}, \ref{lemma:6} and \ref{lemma:7}, where from Lemma \ref{lemma:7} we know that we can set $\epsilon$ to be small enough such that
$$
    \epsilon+\epsilon_1+\epsilon_2+\epsilon_3
    <2\gamma \bar pC
$$
where
$$
    C=1-\frac{1}{\bar pn^2}\sum_{ij}\{P_{ij}(\btheta^*)\tilde\psi_i^*\tilde\psi_j^*\}^{1/2}
    >0. \ \ \ \square
$$

\subsection*{Proof of lemmas}
\begin{lemma}\label{lemma:6}
   {\it Let $\bA\sim\varrho_n\bP(\btheta^*)$ where $\bP(\btheta^*)$ defines a popularity-adjusted block model (PABM), i.e., $P_{ij}(\btheta^*)=\lambda^*_{ic_j^*}\lambda_{jc_i^*}^*$. Furthermore, define $\hat\psi_i=\sum_{j=1}^n A_{ij}/\sqrt{\sum_{ab}A_{ab}}$ and\\ $\tilde\psi^*_i=\sum_{j=1}^nP_{ij}(\btheta^*)/\sqrt{\sum_{ab}P_{ab}(\btheta^*)}$. If}
$$
    X_n=\log \frac{\prod_{ij}(\varrho_n\tilde\psi_i^*\tilde\psi_j^*)^{A_{ij}}}{\prod_{ij}(\hat \psi_i\hat\psi_j)^{A_{ij}}},
$$
{\it then for any $\epsilon>0$,}
$$
    \lim_{n\to\infty}\mathbb P(X_n\leq -\epsilon \varrho_nn^2)=0.
$$ 
\end{lemma}

{\it Proof.}
First, we can replace $\sum_{ij}A_{ij}$ in the denominator of $\hat\psi_i$ with $\bar p\varrho_nn^2$ where
$$
    \bar p=\frac1{n^2}\sum_{ij}P_{ij}(\btheta^*)
$$
as a ``shortcut'' since this term converges much faster than other appropriately scaled random variables. This assumption has been previously made in the literature \citep{yanchenko2025statistical}. Then we can write $X_n$ as
$$
    X_n=\sum_{ij}A_{ij}\log\frac{\varrho_n\tilde\psi_i^*\tilde\psi_j^*}{\hat\psi_i\hat\psi_j}
    =2\sum_{i=1}^n D_i\{\log\varrho_n  \bar P_{i}(\btheta^*)-\log \bar D_i\}.
$$
where
$
    D_i=\sum_{j=1}^nA_{ij}
$
is the block modelof node $i$ for $i\in\{1,\dots,n\}$, $\bar D_i=D_i/n$, and
$
      \bar P_{i}(\btheta^*)=\frac1n\sum_{j=1}^nP_{ij}(\btheta^*)
$
. By the Cauchy-Schwarz Inequality, 
$$
    X_n^2\leq \left(\sum_{i=1}^n D_i^2\right)\left(\sum_{i=1}^n\{\log\varrho_n  \bar P_{i}(\btheta^*)-\log  \bar D_i\}^2\right).
$$
From the proof of Lemma \ref{lemma:3}, we know that
$
    \sum_{i=1}^n D_i^2=\mathcal O_p(\varrho_n^2n^3).
$
By \ref{a:pmin}, we have that $ \bar P_{i}(\btheta^*)>0$ and the probability that $ \bar D_i=0$ is vanishingly small as $n$ increases so by a similar argument to that of Lemma \ref{lemma:3}, we can set $\delta=\min_i\{ \bar D_i,  \bar P_{i}(\btheta^*)\}>0$ such that $\log(x)$ is Lipschitz for $x\in[\delta,\infty)$. Thus, 
$$
    \sum_{i=1}^n\{\log\varrho_n  \bar P_{i}(\btheta^*)-\log \bar D_i\}^2
    \lesssim \sum_{i=1}^n\{\varrho_n  \bar P_{i}(\btheta^*)- \bar D_i\}^2
    = \sum_{i=1}^n(\bar D_i-\mathbb E\bar D_i)^2.
$$
By Hoeffding's Inequality, for all $i$ and any $\varepsilon>0$,
$$
    \mathbb P\{(\bar D_i-\mathbb E\bar D_i)^2\geq \varepsilon\}
    =\mathbb P(|\bar D_i-\mathbb E\bar D_i|\geq \varepsilon^{1/2})
    =\mathbb P(|D_i-\mathbb ED_i|\geq n\varepsilon^{1/2})
    \lesssim \exp(-n\varepsilon).
$$
Taking a union bound over all $n$ nodes, we have
$$
    \mathbb P\left[\bigcup_i\{(\bar D_i-\mathbb E\bar D_i)^2\geq \varepsilon\}\right]
    \lesssim n\exp\left(-n\varepsilon\right)
    \to 0
$$
as $n\to\infty$. Therefore, $(D_i-\mathbb ED_i)^2=o_p(1)$ for all $i$ such that $\sum_{i=1}^n(\bar D_i-\mathbb E\bar D_i)^2=o_p(n)$ and the result we hoped to prove follows immediately by multiplying the rates. $\square$\\

\noindent
\begin{lemma}\label{lemma:7}
    {\it Let $\bA\sim\varrho_n\bP(\btheta_n^*)$ where $P_{ij}(\btheta^*)=\lambda_{ic_j^*}^*\lambda_{jc_i^*}^*$, i.e., generated from a PABM, under \ref{a:cl}. Let}
    $$
        X_n
        =\log \frac{\prod_{ij}\{\varrho_nP_{ij}(\btheta^*)\}^{A_{ij}}}{\prod_{ij}(\varrho_n\tilde\psi_i^*\tilde\psi_j^*)^{A_{ij}}}.
    $$
    {\it where $\tilde\psi^*_i=\sum_{j=1}^n P_{ij}(\btheta^*)/\sqrt{\sum_{ab}P_{ab}(\btheta^*)}$. Then there exists some $\epsilon>0$ such that}
$$
    \lim_{n\to\infty}\mathbb P(X_n\leq \epsilon\varrho_n n^2)=0.
$$
\end{lemma}

\noindent
{\it Proof.} We start by writing out $X_n$ as
$$
    X_n=\sum_{ij}A_{ij}\log\left(\frac{\varrho_nP_{ij}(\btheta^*)}{\varrho_n\tilde\psi_i^*\tilde\psi_j^*}\right)
$$
Similar to the proof of Lemma \ref{lemma:5}, we use a Chernoff bound argument and find
$$
    \mathbb P(X_n\leq \epsilon\varrho_n n^2)
    \leq \exp\left(\bar p\varrho_nn^2\left[\frac{\epsilon}{2\bar p} -  \frac{1}{\bar pn^2}\sum_{ij}P_{ij}(\btheta^*)-\{\tilde\psi_i^*\tilde\psi_j^*P_{ij}(\btheta^*)\}^{1/2}\right]
    \right).
$$
Additionally,
$$
    \sum_{ij}\tilde\psi_i^*\tilde\psi_j^*
    =\frac1{\bar pn^2}\sum_{ijk\ell}P_{ik}(\btheta^*)P_{j\ell}(\btheta^*)
    =\frac1{\bar pn^2}\left(\sum_{ik} P_{ik}(\btheta^*)\right)\left(\sum_{j\ell} P_{j\ell}(\btheta^*)\right)
    =\bar pn^2=\sum_{ij}P_{ij}(\btheta^*).
$$
Thus, by Lemma \ref{lemma:4}
$$
    \frac{1}{\bar pn^2}\sum_{ij}P_{ij}(\btheta^*)-\{\tilde\psi_i^*\tilde\psi_j^*P_{ij}(\btheta^*)\}^{1/2}
    = C
$$
for some $C>0$ where the inequality is strict by \ref{a:cl}. We can set $\epsilon>0$ small enough such that
$$
    \frac{\epsilon}{2\bar p}-C <0\implies \epsilon< 2\bar p C.
$$
Then we have
$$
    \mathbb P(X_n\leq\epsilon\varrho_n n^2)
    \lesssim \exp(-\varrho_nn^2)
    \to 0
$$
as $n\to\infty$ by \ref{a:sparse}. $\square$

\clearpage

\section{Proof of Theorem \ref{thm:sbm}}\label{supp:sbm}
{\it Let $\bA_n\sim \varrho_n\bP(\btheta^*)$ under \ref{a:sparse}. Assume that we are testing 
$$
    H_0:\btheta^*\in\Theta_{0}\text{ vs. }\btheta^*\in\Theta_1
$$
where $\Theta_{0}$ is the parameter space corresponding to the stochastic block model with $K-1>1$ blocks, and $\Theta_1$ corresponds to the stochastic block model with $K$ blocks under \ref{a:K}, \ref{a:pmin} and \ref{a:pbar} where the block probability matrix is also of full rank. If $\btheta^*\in\Theta_1$ such that $\bA_n$ is generated from the alternative model, then there exists some $\epsilon>0$ such that, as $n\to\infty$, }
$$
    \mathbb P(\log E_n>\epsilon\varrho_nn^2)\to 1.
$$

\subsection*{Main proof}
We sketch a proof leaning heavily on the results from \cite{wang2017likelihood, hu2020corrected}. Specifically, define $\bB^*\in(0,1)^{K\times K}$ as the true block-probability matrix and $\bc^*\in\{1,\dots,K\}^n$ as the true block labels which generated the network under the alternative model (with $K$ blocks). Moreover, let $n_k$ be the number of nodes in community $k$, i.e., $n_k=\sum_{i=1}^n \mathds{1}(c_i^*=k)$ for $k\in\{1,\dots,K\}$. \cite{wang2017likelihood} define a ``merging operator'' that merges two blocks in the (correctly specified) block-probability matrix to construct an (underfit) block-probability matrix $\tilde\bB^*\in(0,1)^{K-1\times K-1}$. Without loss of generality, we can assume that we merge blocks $K-1$ and $K$ such that $\tilde\bB^*$ is defined as
\begin{align*}
    \tilde B^*_{k\ell}&=B^*_{k\ell}\text{ for }k,\ell\in\{1,\dots,K-2\}\\
    \tilde B_{k,K-1}^*&=\frac{n_{k}n_{K-1}B^*_{k,K-1}+n_kn_KB_{k,K}^*}{n_kn_{K-1} + n_k n_{K}}\text{ for }k\in\{1,\dots,K-2\}\\
    \tilde B^*_{K-1,K-1}&=\frac{n^2_{K-1}B_{K-1,K-1}^* + 2n_{K-1}n_KB^*_{K-1,K}+ n^2_KB_{K,K}^*}{n^2_{K-1}+2n_{K-1}n_K + n^2_K}.
\end{align*}
The merged labels $\tilde\bc^*$ are then clearly defined as $\tilde c_i^*=c_i^*$ if $c_i^*\in\{1,\dots,K-1\}$ and $\tilde c_i^*=K-1$ if $c_i^*=K$. Moreover, the number of nodes in block $k$ of the underfit model, $\tilde n_k$, is defined as $\tilde n_k=n_k$ for $k\in\{1,\dots,K-2\}$ and $\tilde n_{K-1}=n_{K-1}+n_K$. \\

\noindent
Now, \cite{wang2017likelihood} and \cite{hu2020corrected} show that
$$
    \sup_{\btheta\in\Theta_0} \mathcal L(\bP(\btheta)|\bZ)
    =\max_{\bc\in\{1,\dots,K-1\}^n}\sup_{\bB\in(0,1)^{K-1\times K-1}}\mathcal L(\bP(\bc,\bB)|\bZ)
    =\mathcal L(\bP(\tilde\bc^*,\hat \bB)|\bZ)
$$
where $\tilde\bc^*$ are the merged labels defined above and $\hat\bB$ is defined as
\begin{align*}
    \hat B_{k\ell}
    &=\frac{S_{k\ell}}{ n_k n_\ell}\text{ for } k,\ell\in\{1,\dots,K-2\}\\
    \hat B_{k,K-1}&=\frac{S_{k,K-1}+S_{k,K}}{ n_k n_{K-1}+ n_k n_{K}}\text{ for }k\in\{1,\dots,K-2\}\\
    \hat B_{K-1,K-1}&=\frac{\sum_{k=K-1}^K\sum_{\ell=k}^KS_{k\ell}}{\sum_{k=K-1}^K\sum_{\ell=k}^K n_k n_\ell}
\end{align*}
with
$$
    S_{k\ell}
    =\sum_{ij}Z_{ij}\mathds{1}(c_i^*=k, c_j^*=\ell)\text{ for }k,\ell\in\{1,\dots,K\}.
$$
Note that both authors show this result for the Bernoulli model whereas we use a Poisson approximation in this work. For this reason, we refer to this as a proof ``sketch.''\\

\noindent
We can now follow the basic outline of the proofs for Theorems \ref{thm:er} and \ref{thm:cl}, writing $\log E_n$ as the sum of
\begin{align*}
    B_n&=\log\frac{\prod_{ij} e^{-\gamma \varrho_n P_{ij}(\btheta^*)}}
    {\prod_{ij} e^{-\hat B_{\tilde c_i^*,\tilde c_j^*}}}\\
    C_n&=\log \frac{\prod_{ij}\gamma\varrho_nP_{ij}(\btheta^*)^{Z_{ij}}}{\prod_{ij}(\varrho_n\tilde B^*_{\tilde c_i^*,\tilde c_j^*})^{Z_{ij}}}\\
    D_n&=\log \frac{\prod_{ij}(\gamma\varrho_n\tilde B^*_{\tilde c_i^*,\tilde c_j^*})^{Z_{ij}}}{\prod_{ij}(\hat B_{\tilde c_i^*,\tilde c_j^*})^{Z_{ij}}}\\
    L_n&=\log\frac{\mathcal L(\tfrac{\gamma}{1-\gamma}\bP(\hat\btheta_1)|\bZ)}{\mathcal L(\gamma\varrho_n\bP(\btheta^*)|\bZ)}
\end{align*}
where $\tilde B^*_{k\ell}$ is defined by the merging operator above. The result follows by using Lemmas 1, 3, 8 and 9 where $\epsilon$ is set to be small enough such that
$$
    \epsilon+\epsilon_1+\epsilon_2+\epsilon_3
    <2\gamma \bar pC
$$
where
$$
    C=1-\frac{1}{\bar pn^2}\sum_{ij}\{\tilde B^*_{\tilde c_i^*,\tilde c_j^*}P_{ij}(\btheta^*)\}^{1/2}
    > 0.\ \ \ \square
$$

\subsection*{Proof of lemmas}

\begin{lemma}\label{lemma:8}
    {\it Let $\bA\sim\varrho_n\bP(\btheta^*)$ where $\bP(\btheta^*)$ defines an SBM with $K\geq 2$ blocks. Consider $\{\tilde\bB^*, \tilde\bc^*\}$ and $\hat\bB$ defined above. If}
$$
    X_n
    =\log \frac{\prod_{ij}(\varrho_n\tilde B^*_{\tilde c_i^*,\tilde c_j^*})^{A_{ij}}}{\prod_{ij}(\hat B_{\tilde c_i^*,\tilde c_j^*})^{A_{ij}}},
$$
{\it then for any $\epsilon>0$,}
$$
    \lim_{n\to\infty}\mathbb P(X_n\leq -\epsilon\varrho_n n^2)=0.
$$
\end{lemma}

\noindent
{\it Proof.} 
First,
\begin{align*}
    X_n
    =-\sum_{ij}A_{ij}(\log \hat B_{\tilde c_i^*,\tilde c_j^*}-\log \varrho_n\tilde B^*_{\tilde c_i^*,\tilde c_j^*})
\end{align*}
Then,
$$
    \mathbb P(X_n\leq -\epsilon\varrho_n n^2)
    \leq \mathbb P(|X_n|\geq \epsilon\varrho_n n^2).
$$
We write the (scaled) absolute value as 
$$
    \frac{|X_n|}{\varrho_nn^2}
    \leq \frac1{\varrho_nn^2}\sum_{ij}A_{ij}|\log \hat B_{\tilde c_i^*,\tilde c_j^*}-\log \varrho_n\tilde B^*_{\tilde c_i^*,\tilde c_j^*}|
    = \sum_{k=1}^{K-1}\sum_{\ell=1}^{K-1} \frac{S_{k\ell}}{\varrho_nn^2}\left|\log \frac{\hat B_{k\ell}}{\varrho_n\tilde B^*_{k\ell}}\right|
    \sim \sum_{k=1}^{K-1}\sum_{\ell=1}^{K-1} \frac{\hat B_{k\ell}}{\varrho_n}\left|\log \frac{\hat B_{k\ell}/\varrho_n}{\tilde B^*_{k\ell}}\right|
$$
where
$$
    S_{k\ell}
    =\sum_{ij}A_{ij}\mathds{1}(\tilde c_i^*=k,\tilde c_j^*=\ell).
$$
If we define $f(x)=x|\log (x/\delta)|$ where $\delta=\min_{k,\ell}\{\tilde B^*_{k\ell}\}$ with $\delta>0$ by \ref{a:pmin}, then by the continuous mapping theorem
$$
    \frac{\hat B_{k\ell}}{\varrho_n}\left|\log \frac{\hat B_{k\ell}/\varrho_n}{\tilde B^*_{k\ell}}\right|
    =o_p(1)
$$ 
since \cite{wang2017likelihood} and \cite{hu2020corrected} show that $\hat B_{k\ell}/\varrho_n\stackrel{p}{\to}\tilde B_{k\ell}$ for all $(k,\ell)$.
The boundedness of $x\log(x/\tilde B^*_{k\ell})$ on $[\delta,1]$ also allows us to use the Dominated Convergence Theorem  to show
\begin{align*}
    \lim_{n\to\infty}\frac{X_n}{\varrho_nn^2}
    &\sim \lim_{n\to\infty} \sum_{k=1}^{K-1}\sum_{\ell=1}^{K-1} \frac{\hat B_{k\ell}}{\varrho_n}\left|\log \frac{\hat B_{k\ell}/\varrho_n}{\tilde B^*_{k\ell}}\right|\\
    &\sim  \sum_{k=1}^{K-1}\sum_{\ell=1}^{K-1}\lim_{n\to\infty} \frac{\hat B_{k\ell}}{\varrho_n}\left|\log \frac{\hat B_{k\ell}/\varrho_n}{\tilde B^*_{k\ell}}\right|\\
    &=\sum_{k=1}^{K-1}\sum_{\ell=1}^{K-1}o_p(1)\\
    &=o_p(1) 
\end{align*}
since $K$ does not depend on $n$. Thus, for any $\epsilon>0$,
$$
    \mathbb P(X_n\leq-\epsilon\varrho_n n^2)
    \leq \mathbb P(|X_n|\geq \epsilon\varrho_n n^2) 
    \to 0
$$
as $n\to\infty$. $\square$

\begin{lemma}\label{lemma:9}
    {\it Let $\bA\sim\varrho_n\bP(\btheta_n^*)$ where $P_{ij}(\btheta^*)=B_{c_i^*,c_j^*}^*$, i.e., generated from an SBM with $K\geq 2$ blocks where the block probability matrix is of full rank. Let}
    $$
        X_n
        =\log \frac{\prod_{ij}P_{ij}(\btheta^*)^{A_{ij}}}{\prod_{ij}(\tilde B^*_{\tilde c_i^*,\tilde c_j^*})^{A_{ij}}}.
    $$
    {\it where $\tilde \bB$ is defined by the merging operator above. Then there exists some $\epsilon>0$ such that}
$$
    \lim_{n\to\infty}\mathbb P(X_n\leq \epsilon\varrho_n n^2)=0.
$$    
\end{lemma}

\noindent
{\it Proof.} We start by writing out $X_n$ as
$$
    X_n=\sum_{ij}A_{ij}\log\left(\frac{P_{ij}(\btheta^*)}{\tilde B^*_{\tilde c_i^*,\tilde c_j^*}}\right)
$$
Similar to the proof of Lemma \ref{lemma:5}, we use a Chernoff bound argument and find
$$
    \mathbb P(X_n\leq \epsilon\varrho_n n^2)
    \leq \exp\left(\bar p\varrho_nn^2\left[\frac{\epsilon}{2\bar p} -  \frac{1}{\bar pn^2}\sum_{ij}P_{ij}(\btheta^*)-\{\tilde B^*_{\tilde c_i^*,\tilde c_j^*}P_{ij}(\btheta^*)\}^{1/2}\right]
    \right).
$$
Additionally, 
\begin{align*}
    \sum_{ij}\tilde B^*_{\tilde c_i^*,\tilde c_j^*}
    &=\sum_{k=1}^{K-1}\sum_{\ell=1}^{K-1} n_kn_\ell\tilde B^*_{k\ell}\\
    &=\sum_{k=1}^{K-2}\sum_{\ell=1}^{K-2} n_kn_\ell B^*_{k\ell}
    +\sum_{k=1}^{K-1}(n_kn_{K-1}B_{k,K-1}^*+n_kn_KB_{k,K}^*) + \sum_{k=1}^{K-1}(n_kn_{K-1}B_{K-1,k}^*+n_kn_KB_{K,k}^*)\\
    &+n^2_{K-1}B_{K-1,K-1}^*+ 2n_{K-1}n_KB^*_{K-1,K}+ n^2_KB_{K,K}^*\\
    &=\sum_{ij} P_{ij}(\btheta^*)=\bar p n^2.
\end{align*}
Thus, by Lemma \ref{lemma:4}
$$
    \frac{1}{\bar pn^2}\sum_{ij}P_{ij}(\btheta^*)-\{\tilde B^*_{\tilde c_i^*,\tilde c_j^*}P_{ij}(\btheta^*)\}^{1/2}
    = C
$$
for some $C>0$ where the inequality is strict since $\bB^*$ is full rank. In other words, $\bB^*$ cannot be collapsed to a smaller model such that it is truly distinct from $\tilde \bB^*$. We can set $\epsilon>0$ small enough such that
$$
    \frac{\epsilon}{2\bar p}-C <0\implies \epsilon< 2\bar p C.
$$
Then we have
$$
    \mathbb P(X_n\leq\epsilon\varrho_n n^2)
    \lesssim \exp(-\varrho_nn^2)
    \to 0
$$
as $n\to\infty$ by \ref{a:sparse}. $\square$

\clearpage

\section{Additional simulation results: DCBM versus PABM}\label{supp:sim}
We provide an additional simulation setting to demonstrate the flexibility of the proposed method. Similar to the SBM vs.~DCBM test, we test two models which both exhibit community structure but have different levels of flexibility. Namely, we test the null hypothesis that the network was generated from a DCBM against the alternative model of a PABM. To generate PABM networks, we follow a similar approach as in \cite{senguptapabm}. For node $i\in\{1,\dots,n\}$ and community labels $\bc\in\{1,\dots,K\}^n$,
$$
    \lambda_{ij}
    =
    \begin{cases}
        a/\sqrt{\beta/(1+\beta)}&c_i=c_j\\
        b/\sqrt{1/(1+\beta)}&c_i\neq c_j.
    \end{cases}
$$
We fix $K=2$ and assign $\delta=0.25$ proportion of the nodes to community one and the remaining $1-\delta$ nodes to community two. In each community we select half of the nodes as group 1 and the other half as group 2. We then set $a=\frac12(1+\nu)$ and $b=\frac12(1-\nu)$ for nodes in group 1 and $a=\frac12(1-\nu)$ and $b=\frac12(1+\nu)$ for nodes in group 2. In words,
for a group 1 and group 2 node in community one, the group 1 node is more popular in community one while the group 2 node is more popular in community two. When $\nu=0$, this reduces to the null hypothesis of the network coming from a DCBM. We set $\beta=3$ and vary $\nu\in\{0.10, 0.11, \dots, 0.16\}$. Since $\nu>0$, the networks are generated from a proper PABM so the test should reject the null hypothesis. Neither ECV nor NETCROP are described for the PABM, so we only consider the proposed method with $\gamma\in\{0.2, 0.4, 0.5, 0.6\}$. The results are in Figure \ref{fig:sbm_pabm} averaged over 100 MC repetitions. All values of $\gamma$ yield a non-decreasing power with $\gamma=0.4$ yielding the best results.

\begin{figure}
    \centering
    \includegraphics[width=0.55\linewidth]{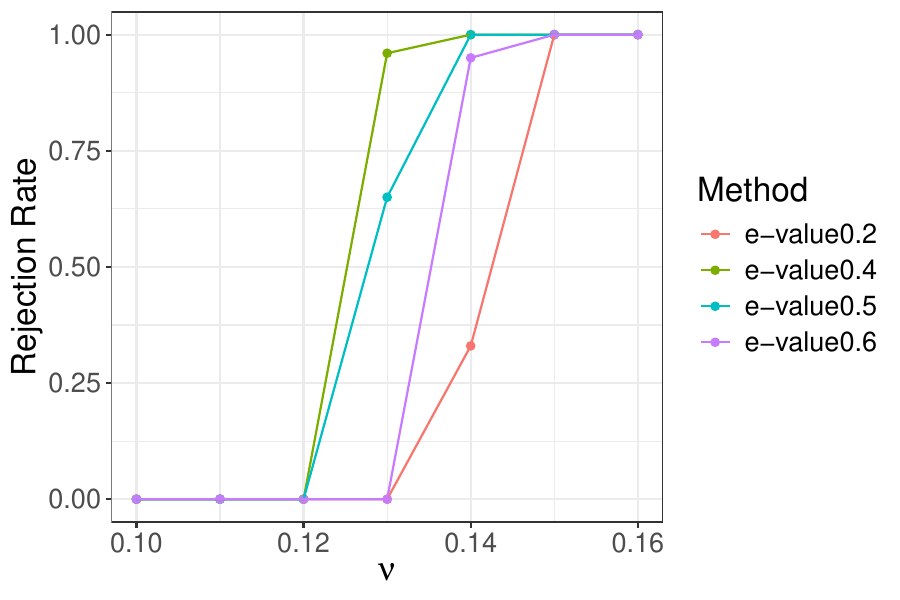}
    \caption{$H_0: \nu=0$ (DCBM) vs. $H_1: \nu>0$ (PABM) with increasing $\nu\in[0.10,0.16]$ and $K=2$. All networks are generated under the alternative model.}
    \label{fig:sbm_pabm}
\end{figure}

\end{document}